%% file: main.tex
\documentclass{eptcs}
\providecommand{\event}{ICE 2018} % Name of the event you are submitting to

\usepackage{breakurl}             % Not needed if you use pdflatex only.
\usepackage{underscore}           % Only needed if you use pdflatex.

\input{styles}
\input{macros_util}

\input{ggmacros}

\newtheorem{lemma}{Lemma}

\title{Realisability of Pomsets via Communicating Automata\thanks{
    Research partly supported by the EU H2020-RISE-2017 project BehAPI
    and the EU COST Action IC1405.
    \newline
    The authors thank the anonymous reviewers for their comments and
    the interesting discussions on the forum of ICE18 .
  }
}

\author{
  Roberto Guanciale \institute{KTH Royal Institute of Technology, Sweden} \email{robertog@kth.se}
  \and
  Emilio Tuosto \institute{Department of Informatics, University of Leicester, UK} \email{emilio@le.ac.uk}
}

\def\titlerunning{Realisability of Pomsets}
\def\authorrunning{R. Guanciale \& E. Tuosto}

\newcommand{\fnrestriction}[1]{\lvert_{#1}}
\newcommand{\size}[1]{\mid #1 \mid}
\newcommand{\projword}[2]{{#1}\!\downharpoonleft_{#2}}
\newcommandx{\aQfinal}[1][1=,usedefault=@]{
  {\ifempty{#1}{F}{F_{#1}}}
}
\newcommand{\CC}[1][2]{\textsf{\textbf{CC{#1}}}}
\newcommand{\CCP}[1][2]{\CC[{#1}-POM]}
\newcommand{\pref}[1][\rlang]{{\colorSbj{\textsf{pref}\ifempty{#1}{}{(#1)}}}}
\newcommand{\apomMtrans}[1]{\xrightarrow{#1}}
\newcommand{\rsubtype}{\sqsubseteq}
\def\colorSbj{\color{NavyBlue}}
\newcommandx{\countEvents}[3][1=\apom,2=\ae,3=\al,usedefault=@]{\mathsf{card}^{#1}_{#3}(#2)}

\begin{document}
\maketitle

\begin{abstract}
  Pomsets are a model of concurrent computations introduced by
  Pratt.
  They can provide a syntax-oblivious description of semantics of
  coordination models based on asynchronous message-passing, such as
  Message Sequence Charts (MSCs).
  In this paper, we study conditions that ensure a specification
  expressed as a set of pomsets can be faithfully realised via
  communicating automata.

  Our main contributions are (i) the definition of a realisability
  condition accounting for termination soundness, (ii) conditions for
  global specifications with \quo{multi-threaded} participants, and
  (iii) the definition of realisability conditions that can be decided
  directly over pomsets.
  A positive by-product of our approach is the efficiency gain
  in the verification of the realisability conditions
  obtained when restricting to specific classes of choreographies
  characterisable in term of behavioural types.
\end{abstract}

\section{Introduction}
\label{sec:intro}
\input{intro}

\section{Pomsets and message-sequence charts}
\label{sec:msc}
\input{msc}

\section{Realisability and termination soundness
of pomsets}
\label{sec:realisability}
\input{realisability}

\section{Pomset based verification conditions}
\label{sec:pomsets}
\input{pomsets}

\section{Discussion on the pomset based conditions}
\label{sec:implementation}
\input{implementation}

\section{Related work}
\label{sec:related}
\input{related}

\section{Concluding remarks}
\label{sec:conc}

\input{conc}

\bibliographystyle{eptcs}
\bibliography{bib}

\end{document}

%% file: styles.tex
\usepackage[utf8]{inputenc}
\usepackage{xargs}

\usepackage[nomargin,multiuser,inline,draft]{fixme} 
\fxusetheme{color}
\FXRegisterAuthor{eM}{em}{{\color{orange} \underline{eM} }}

\usepackage{multicol}

\usepackage{mathtools}
\usepackage{amsmath}
\usepackage{amsthm}

\usepackage{amssymb}
\usepackage{environ}                % Provides NewEnviron w/ eager expansion
\usepackage{etex,etoolbox}          % Needed by theorem-appendix.tex
\usepackage{graphicx}               % For \scalebox etc.

\usepackage[protrusion=true,expansion=true]{microtype} % Better typography
\usepackage{multirow}               % Table multirow/multicolumn support
\usepackage{microtype}            % Trying to avoid underfull hboxes
\usepackage{hyperref}
\usepackage[capitalise]{cleveref}
\usepackage{nicefrac}
\usepackage{proof}                  % Logic proofs
\usepackage{bussproofs}
\EnableBpAbbreviations

\usepackage{subcaption}             % Support for subfigures
\usepackage{thmtools, thm-restate, thm-autoref}  % Advanced theorem handling
\usepackage{xifthen}                % for conditional commands
\usepackage{appendix}               % for correct numeration of statements

\usepackage{stmaryrd}
\usepackage{pdfsync}
\usepackage{pslatex}

\usepackage[usenames,dvipsnames,svgnames,table]{xcolor} % must be loaded first
\usepackage{tikz}                   % Graphs
\usetikzlibrary{shadows,arrows,shapes,automata,positioning,decorations.markings}

\usepackage{todonotes}
\DeclareMathSizes{10}{10}{6}{6}

%% file: macros_util.tex
\newif\ifemi
\usepackage{bm}
\usepackage[draft]{fixme}
\fxusetheme{color}
\usepackage[normalem]{ulem} % underline command breaks over line ends

\usepackage{xifthen}        % for conditional commands
\newcommand{\ifempty}[3]{%
  \ifthenelse{\isempty{#1}}{#2}{#3}%
}

\newcommand{\mkfun}[4][\colorFun]{
  \newcommand{#2}[1][#4]{
    {#1\textsf{#3}}
    \ifempty{##1}{}{
      \big({##1}\big)}
  }
}

\newcommand{\mkuop}[4][\colorFun]{
  \newcommand{#2}[1][#4]{
    {#1\textsf{#3}}
    \ifempty{##1}{}{
      \, {##1}}
  }
}

\newcommand{\hidden}[1]{}

\newcommand{\cf}[2]{
  \fontsize{#1}{#1}{\selectfont{#2}}
}
\ifemi
\usepackage{showlabels}

\newcommand{\emi}[2]{
  \marginpar{\fcolorbox{red}{shadecolor}{\cf{#1}{{#2}}}}
}
\newcommand{\emic}[2]{\par
  \fcolorbox{red}{shadecolor}{\parbox{\linewidth}{ 
      \color{gray}
      \begin{description}
      \item[{\color{blue} #2}]{\sf #1}
      \end{description}}}
}
\else
\newcommand{\emi}[2]{}
\newcommand{\emic}[2]{}{}
\fi

\newcommand{\sst}{\;\big|\;}
\newcommand{\qst}{\;\colon\;} %such that

\newcommand{\conf}[1]{\ensuremath{\langle {#1} \rangle}}

\newcommand{\langof}[1]{\mathbb{L}(#1)}

\newcommand{\qqand}{\qquad\text{and}\qquad}

\newcommand{\upd}[3]{{#1}[{#2} \mapsto {#3}]}

{\bfseries}{\rmfamily}
\newtheorem{definition}{Definition}{\bfseries}{\rmfamily}
\newtheorem{theorem}{Theorem}{\bfseries}{\rmfamily}
{\bfseries}{\rmfamily}

\newcommand{\quo}[1]{\lq\lq {#1}\rq\rq}
\def\finex{{\unskip\nobreak\hfil
\penalty50\hskip1em\null\nobreak\hfil$\diamond$
\parfillskip=0pt\finalhyphendemerits=0\endgraf}}

\definecolor{shadecolor}{rgb}{1,0.99,0.9}
\definecolor{bg}{rgb}{0.95,0.95,0.95}
\newcommand{\newgreen}{green!50!blue!100}

%% file: ggmacros.tex
\usepackage{listings}
\usepackage{mfirstuc}
\newcommand{\fillcolor}{orange!5}
\def\colorPtp{\color{blue}}
\def\colorFun{\color{NavyBlue}}

\def\colorOp{\color{OliveGreen}}
\def\colorNode{\color{cyan}}
\def\colorR{\color{OliveGreen}}
\def\colorE{\color{orange}}

\def\colorMsg{\color{BrickRed}}

\newcommand{\msg}[1][m]{\mathsf{\colorMsg{#1}}}
\newcommand{\msgset}{\mathcal{M}}
\newcommand{\chset}{\mathcal{C}}

\newcommand{\lset}{\mathcal{L}}

\newcommand{\ptp}[1][A]{\ensuremath{\mathsf{\colorPtp{\capitalisewords{#1}}}}}

\newcommand{\p}{\ptp}
\newcommand{\q}{{\ptp[B]}}
\newcommand{\s}{{\ptp[S]}}

\newcommandx{\ggcommon}[3][1=\ptp,2={\aH},3={\aH'},usedefault=@]{f_{#1}}
\newcommandx{\opair}[2][1={\ae},2={\ae'},usedefault=@]{\conf{{#1},{#2}}}
\newcommandx{\hopair}[2][1={\aE},2={\aE'},usedefault=@]{\llparenthesis\, {#1},{#2}\, \rrparenthesis}
\newcommandx{\wf}[2][1={\aG},2={\aG'},usedefault=@]{wf({#1}, {#2})}
\newcommandx{\wb}[2][1={\aG},2={\aG'},usedefault=@]{wb({#1}, {#2})}
\newcommandx{\ws}[2][1={\aG},2={\aG'},usedefault=@]{ws({#1}, {#2})}
\newcommandx{\widx}[2][1={\aW},2={i},usedefault=@]{{#1}[{#2}]}
\newcommandx{\outop}[2][1=\gname,2={}]{{\colorOp{!}}^{{#1}{#2}}}
\newcommandx{\inop}[2][1=\gname,2={}]{{\colorOp{?}}^{{#1}{#2}}}
\newcommandx{\aout}[5][1={\p},2={\q},3={},4=\msg,5={},usedefault=@]{
  \achan[#1][#2] \outop[{#3}] {#4}{#5}
}
\newcommandx{\ain}[5][1={\p},2={\q},3={},4=\msg,5={},usedefault=@]{
  \achan[#1][#2] \inop[{#3}] {#4}{#5}
}
\newcommandx{\adep}[1][1={}]{
  \conf{ \aout[@][@][@][@][{#1}], \ain[@][@][@][@][{#1}]}
}

\newcommandx{\hproj}[2][1=\aH, 2=\ptp, usedefault=@]{
  \ifempty{#1}{}{{#1}}\ifempty{#2}{}{{^{\scriptscriptstyle @{#2}}}}
}
\newcommandx{\eproj}[2][1=\aE,2=\ptp, usedefault=@]{
  {{#1}}\ifempty{#2}{}{{^{\scriptscriptstyle @{#2}}}}
}

\newcommand{\apom}{r}

\newcommand{\alf}{\lambda}
\newcommand{\projpom}[2]{{#1}\!\!\downharpoonright_{#2}}

\newcommand{\eset}{\mathcal{E}}
\newcommand{\aR}[1][R]{{\colorR{#1}}}

\newcommand{\aConf}{s}
\newcommand{\alfof}[1]{\alf_{#1}}

\newcommand{\esetof}[1]{\eset_{#1}}
\newcommand{\leqof}[1]{\leq_{#1}}
\newcommandx{\detM}[1][1=\aCM,usedefault=@]{\Delta({#1})}
\newcommandx{\cm}[2][1=\ptp, 2=\aM]{{#2}_{#1}}
\newcommandx{\achan}[2][1=A,2=B,usedefault=@]{{\ptp[#1]\,\ptp[#2]}}
\newcommand{\ptpset}{\mathcal{\colorPtp{P}}}

\newcommand{\pset}{\ptpset}

\newcommand{\oact}{\outop[]}
\newcommand{\iact}{\inop[]}
\newcommand{\tset}{\to}

\newcommand{\csconf}[2]{\conf{\vec{#1} \ ; \ \vec{#2}}}

\newcommand{\TRANSS}[1]{{\xRightarrow{\raisebox{-.3ex}[0pt][0pt]{$\scriptstyle #1$} }}}

\newcommand{\trans}[2][{}]{\,\xrightarrow{#2}_{#1}\,}

\newcommandx{\acfsmout}[3][1=A,2=B,3=m,usedefault=@]{\achan[{\ptp[#1]}][{\ptp[#2]}] \oact {\msg[{#3}]}}
\newcommandx{\acfsmin}[3][1=A,2=B,3=m,usedefault=@]{\achan[{\ptp[#1]}][{\ptp[#2]}] \iact {\msg[{#3}]}}
\newcommandx{\fsaout}[2][1={\p},2={},usedefault=@]{
  \ptp[#1] \ \outop[]\ \msg[{#2}]
}
\newcommandx{\fsain}[2][1={\p},2={},usedefault=@]{
  \ptp[#1] \ \inop[]\ \msg[{#2}]
}

\makeatletter
\newcommand{\linenumfontsize}{\@setfontsize{\linenumfontsize}{3pt}{3pt}}
\makeatother
\lstdefinelanguage{sys}{
	commentstyle=\color{Gray},
	morecomment=[s]{[}{]},
	keywords=[0]{system,of,do,end},	keywordstyle=\color{orange}\bfseries,
}

\newcommand{\aG}{\mathsf{G}}

\newcommand{\gseqop}{{\colorOp ;}\,}
\newcommand{\gparop}{{\colorOp \ |\ }}
\newcommand{\gchoop}{{\colorOp \ +\ }}
\newcommand{\grecop}{{\colorOp *}}
\newcommand{\grecopp}{{\colorOp{@}}}
\newcommand{\gname}[1][i]{{\colorNode{\scriptstyle\textsf{#1}}}}
\newcommandx{\nmerge}[2][1={i},2={},usedefault=@]{
  \ifempty{#2}{
    \ifempty{#1}{\mu}{-\gname[{#1}]}
  }{-{#2}}
}

\mkfun{\esbj}{sbj}{\ae}
\newcommand{\esubject}{\esbj}
\makeatletter%
\@ifclassloaded{exam-paper}%
  {}%
  {\makeatletter%
    \@ifclassloaded{test}%
    {}%
    \makeatother%
  }
\makeatother%

\newcommandx{\gnode}[2][1=i,2=\gint,usedefault=@]{
  {\ifempty{#1}{}{\colorNode{\gname[{#1}].}}} {#2}
}

\newcommandx{\gint}[4][1=i,2=\ptp,3=\msg,4=\q,usedefault=@]{
  \ptp[#2] {\colorOp \xrightarrow{\scriptscriptstyle\gname[#1]}} \ptp[#4] \colon {\msg[{#3}]}
}
\newcommandx{\gout}[4][1=\gname,2=\ptp,3=\msg,4={\ptp[C]},usedefault=@]{
  \achan[{#2}][{#4}] {\colorOp {\colorOp{!}}} {\msg[{#3}]}
}
\newcommandx{\gin}[4][1=\gname,2=\ptp,3=\msg,4={\ptp[C]},usedefault=@]{
  \achan[{#2}][{#4}] {\colorOp {\colorOp{?}}} {\msg[{#3}]}
}
\newcommandx{\gseq}[3][1=i,2={\aG},3={\aG'},usedefault=@]{
  \gnode[{#1}][{#2} \gseqop {#3}]
}
\newcommandx{\gpar}[3][1=i,2={\aG},3={\aG'},usedefault=@]{
  \gnode[{#1}][\ifempty{#1}{{#2} \gparop {#3}}{({#2} \gparop {#3})}]
}
\newcommandx{\gcho}[3][1=i,2={\aG},3={\aG'},usedefault=@]{
  \gnode[{#1}][\ifempty{#1}{{#2} \gchoop {#3}}{\big({#2} \gchoop {#3}\big)}]
}
\newcommandx{\gchov}[3][1=i,2={\aG},3={\aG'},usedefault=@]{
  \gnode[{#1}][\left(
  \begin{array}l
    \ifempty{#1}{{#2} \\ \gchoop \\ {#3}}{\!\!{#2} \\ \gchoop \\ {#3}}
  \end{array}\right)
  ]
}
\newcommandx{\grec}[3][1=i,2={\aG},3={\p},usedefault=@]{
  \gnode[{#1}][\ifempty{#1}{\grecop {#2} \grecopp {#3}}{\big(\grecop {#2} \grecopp {#3}\big)}]
}

\newcommand{\getcentroid}[2]{
    \coordinate (tmpgatecoord) at (0,0);
    \foreach \n [count=\i] in {#1}{
      \path (\n);
      \coordinate (tmpgatecoord) at ($(tmpgatecoord) + (\n)$);
      \coordinate (#2) at ($1/\i*(tmpgatecoord)$);
    }
}

\tikzset{
  hgsem/.style={
    draw,
    node distance=2cm and 1cm,
    transform shape,
    smooth,
    every node/.style = {font=\sffamily\bfseries}
  }
}

\tikzset{
  hgstyle/.style={
    src color={#1},
    tgt color={#1},
    centroid color={#1},
    centroid label={#1},
    centroid name={#1},
    centroid radius={#1},
    centroid ratio={#1},
    xoffset={#1},
    yoffset={#1},
    xsrcoffset={#1},
    ysrcoffset={#1},
    xtgtoffset={#1},
    ytgtoffset={#1},
    font={#1},
    centroid angle={#1},
    centroid tolerance={#1}
  },
  src color/.store in = \hgsrccol,
  tgt color/.store in = \hgtgtcol,
  centroid color/.store in =\hgfillcolor,
  centroid label/.store in =\hglabel,
  centroid name/.store in =\hgname,
  centroid radius/.store in = \hgradius,
  centroid ratio/.store in = \hgratio,
  xoffset/.store in =\hgxoffset,
  yoffset/.store in =\hgyoffset,
  xsrcoffset/.store in =\hgxsrcoffset,
  ysrcoffset/.store in =\hgysrcoffset,
  xtgtoffset/.store in =\hgxtgtoffset,
  ytgtoffset/.store in =\hgytgtoffset,
  centroid angle/.store in =\hgangle,
  centroid tolerance/.store in =\hgtolerance,
  src color = black,
  tgt color = black,
  centroid color = orange!40,
  centroid label={},
  centroid name={dummycentroid},
  centroid radius = .7pt,
  centroid ratio = .35,
  xoffset = 0,
  yoffset = 0,
  xsrcoffset = 0,
  ysrcoffset = 0,
  xtgtoffset = 0,
  ytgtoffset = 0,
  font=\sffamily\scriptsize,
  centroid angle=0,
  centroid tolerance=10pt
}

\newcommandx{\mkhg}[5][1={},4={},5={},usedefault=@]{
  \begingroup
  \tikzset{#1}
  \StrCount{#2,}{,}[\l] % from package xxstring
  \StrCount{#3,}{,}[\m] % from package xxstring
  \ifthenelse{\l = 1 \AND \m = 1}{
    \ifempty{#4}{
      \ifempty{#5}{
        \path[hgsem, ->, >=stealth', shorten >=1pt] (#2) -- (#3);
      }{
        \path[hgsem, ->, >=stealth', shorten >=1pt] (#2) #5 (#3);
      }
    }{
      \ifempty{#5}{
        \path[hgsem, ->, >=stealth', shorten >=1pt, #4] (#2) -- (#3);
      }{
        \path[hgsem, ->, >=stealth', shorten >=1pt, #4] (#2) #5 (#3);
      }
    }
  }{
    \coordinate (srcoffset) at (\hgxsrcoffset,\hgysrcoffset);
    \coordinate (tgtoffset) at (\hgxtgtoffset,\hgytgtoffset);
    \getcentroid{#2}{srccentroid};
    \getcentroid{#3}{tgtcentroid};
    \node[label={left:\hglabel}] (\hgname) at ($(srccentroid)!{1-\hgratio}!\hgangle:(tgtcentroid) + (\hgxoffset,\hgyoffset)$) {};
    \pgfgetlastxy \xc \yc;
    \pgfmathtruncatemacro{\xcontrol}{\xc};
    \pgfmathtruncatemacro{\ycontrol}{\yc};
    \foreach \n in {#2}{
      \path (\n);
      \pgfgetlastxy \xntmp \yntmp;
      \pgfmathtruncatemacro{\xn}{\xntmp};
      \pgfmathtruncatemacro{\yn}{\yntmp};
      \pgfmathsetmacro\xtmpdiff{abs(\xn - \xcontrol + \hgxsrcoffset)};
      \pgfmathsetmacro\ytmpdiff{abs(\yn - \ycontrol + \hgytgtoffset)};
      \ifdim \xtmpdiff pt > \hgtolerance
      \ifempty{#4}{
        \path[hgsem, \hgsrccol] (\n) .. controls ($(srccentroid.center) + (srcoffset)$) .. (\hgname.center);
      }{
        \path[hgsem, \hgsrccol] (\n) .. controls ($(srccentroid.center) + (srcoffset)$) .. (\hgname.center);
      }
      \else
      \ifempty{#4}{
        \path[hgsem, \hgsrccol] (\n) -- (\hgname.center);
      }{
        \path[hgsem, \hgsrccol, #4] (\n) -- (\hgname.center);
      }
      \fi
    }
    \foreach \n in {#3}{
      \path (\n);
      \pgfgetlastxy \xntmp \yntmp;
      \pgfmathtruncatemacro{\xn}{\xntmp};
      \pgfmathtruncatemacro{\yn}{\yntmp};
      \pgfmathsetmacro\xtmpdiff{abs(\xn - \xcontrol)};
      \pgfmathsetmacro\ytmpdiff{abs(\yn - \ycontrol)};
      \ifdim \xtmpdiff pt > \hgtolerance
      \ifempty{#4}{
        \path[hgsem, ->, >=stealth', shorten >=1pt, \hgtgtcol] (\hgname.center) .. controls (tgtcentroid.center) and ($(tgtcentroid.center) + (tgtoffset)$) .. (\n);
      }{
        \path[hgsem, ->, >=stealth', shorten >=1pt, \hgtgtcol,#4] (\hgname.center) .. controls (tgtcentroid.center) and ($(tgtcentroid.center) + (tgtoffset)$) .. (\n);
      }
      \else
      \ifempty{#4}{
        \path[hgsem, ->, >=stealth', shorten >=1pt, \hgtgtcol] (\hgname.center) --  (\n);
      }{
        \path[hgsem, ->, >=stealth', shorten >=1pt, \hgtgtcol] (\hgname.center) --  (\n);
      }
      \fi
    }
    \fill[\hgfillcolor] (\hgname) circle [radius=\hgradius];
  }
  \endgroup
}

\newcommandx{\hgordeq}[1][1={\aH},usedefault=@]{\sqsubseteq_{#1}}
\newcommandx{\gintsem}[4][4=.5]{
  \tikz[hgsem,scale=#4,every node/.style={font=\scriptsize}]{
    \node (out) {$\aout[{#1}][{#2}][][{\msg[{#3}]}]$};
    \node[below = 20pt of out] (in) {$\ain[{#1}][{#2}][][{\msg[{#3}]}]$};
    \mkhg{out}{in};
  }
}

\newcommandx{\gsem}[2][1={\aG},2={},usedefault=@]{[\![ {#1} ]\!]_{#2}}
\newcommandx{\rbot}{\text{undef}}
\newcommandx{\rtrs}[1][1={\aH},usedefault=@]{{#1}^{\star}}
\newcommandx{\gord}[1][1={\aG},usedefault=@]{\leq_{#1}}
\newcommandx{\gordeq}[1][1={\aG},usedefault=@]{\leq_{#1}}
\mkfun{\cause}{cs}{}
\mkfun{\effect}{ef}{}
\newcommandx{\aW}{w}
\newcommandx{\rlang}{\mathsf{L}}
\newcommand{\gfun}[1]{\ensuremath{\mathsf{\colorFun #1}}}
\mkfun{\eact}{\gfun{act}}{}
\mkfun{\enode}{\gfun{cp}}{}

\mkuop{\rmax}{\gfun{max}}{\aH}
\mkuop{\rmin}{\gfun{min}}{\aH}
\mkuop{\rMAX}{\gfun{lst}}{\aH}
\mkuop{\rMIN}{\gfun{fst}}{\aH}

\newcommandx{\rseq}[2][1=\aG,2={\aG'},usedefault=@]{\gfun{seq}({#1},{#2})}
\newcommandx{\rpar}[2][1=\aG,2={\aG'},usedefault=@]{\gfun{par}({#1},{#2})}

\newcommandx{\gproj}[2][1=\aG,2=\ptp]{{#1}\downarrow_{#2}}
\newcommandx{\cinit}[1][1={\aQzero},usedefault=@]{{#1}}
\newcommandx{\cfinal}[1][1={q_e},usedefault=@]{{#1}}

\newcommandx{\geproj}[4][1=\aG,2=\ptp,3=\cinit,4=\cfinal,usedefault=@]{
  {#1}\downarrow_{#2}^{{#3},{#4}}
}

\newcommand*{\StrikeThruDistance}{0.15cm}%
\tikzset{strike thru arrow/.style={
    decoration={markings, mark=at position 0.5 with {
        \draw [blue, thick,-] 
            ++ (-\StrikeThruDistance,-\StrikeThruDistance) 
            -- ( \StrikeThruDistance, \StrikeThruDistance);}
    },
    postaction={decorate},
}}

\newcommandx{\ich}[1][1={\aG},usedefault=@]{{#1}^{\oplus}}
\newcommandx{\ichedges}[2][1={\aG},2={\gname},usedefault=@]{{#1}^{\oplus}({#2})}
\newcommandx{\parts}[1]{2^{#1}}
\newcommandx{\actch}{c}
\newcommandx{\soundactch}[2][1={\aG},2={\actch},usedefault=@]{{#1} \,\circledR\, {#2}}
\newcommandx{\rOnActch}[2][1={\aG},2={\actch},usedefault=@]{{#1} \setminus {#2}}
\newcommandx{\rOnActchClean}[2][1={\aG},2={\actch},usedefault=@]{{#1} \circledR {#2}}
\newcommandx{\rAllEvents}[1][1={\aG},usedefault=@]{\mathit{dom}(#1)}

\newcommand{\AV}{\mathcal{V}}
\newcommand{\aH}{H}

\newcommandx{\hgvertex}[2][1=\al,2=\gname,usedefault=@]{{#1}_{\textcolor{red}{[{#2}]}}}
\newcommand{\aE}{{\colorE E}}
\renewcommand{\ae}[1][e]{{\colorE{#1}}}

\newcommand{\al}[1][l]{{\colorE{#1}}}
\newcommandx{\hyedge}[1]{\{#1\}}

\newcommandx{\rdiv}[2][1=\gcho,2=\ptp,usedefault=@]{
  \gfun{div}_{#2}(#1)
}

\newcommandx{\rrdiv}[5][1={\aG},2={\aG'},3={\AV},4={\AV'},5=\ptp,usedefault=@]{
  \gfun{div}^{#3,#4}_{#5}(#1,#2)
}
\newcommandx{\pdiv}[3][1={\apom_1},2={\apom_2},3={\apom},usedefault=@]{
  \gfun{div}_{#3}(#1,#2)
}
\newcommandx{\pfork}[3][1={\apom_1},2={\apom_2},3={\apom},usedefault=@]{
  \gfun{fork}_{#3}(#1,#2)
}

\newcommandx{\mkint}[6][3=i,4=\p,5=\msg,6=\q,usedefault=@]{
  \node[bblock,{#1}] (#2) {$\gint[#3][#4][#5][#6]$};
}

\newcommandx{\mkgraph}[3][1=.5cm]{
  \node[source,above = #1 of {#2}] (src#2) {};
  \node[sink,below  = #1 of {#3}] (sink#3) {};
  \path[line] (src#2) -- (#2);
  \path[line] (#3) -- (sink#3);
}

\newcommandx{\mkloop}[4][1=.5,2=1.5]{
  \node[ogate,above = #1 of {#3}] (entry#3) {};
  \pgfgetlastxy \xentry \yentry;
  \pgfmathtruncatemacro{\xentryrounded}{\xentry};
  \node[ogate,below  = #1 of {#4}] (exit#4) {};
  \pgfgetlastxy \xexit \yexit;
  \pgfmathtruncatemacro{\xexitrounded}{\xexit};
  \path[line] (entry#3) -- (#3);
  \path[line] (#4) -- (exit#4);
  \pgfmathsetmacro\tmpdiff{abs(\xentryrounded - \xexitrounded)}
  \path[line] (exit#4) -|  ($(exit#4)+(\tmpdiff,0)+(#2,0)$) |- (entry#3);
}

\newcommandx{\mkfork}[4][2=gatenode,3=i,4=.6,usedefault=@]{
  \mkgatebegin{#1}[{\gname[#3]}][agate][#4]{#2}
}

\newcommandx{\mkbranch}[4][2=gatenode,3=i,4=.6,usedefault=@]{
  \mkgatebegin{#1}[{\gname[#3]}][ogate][#4]{#2}
}

\newcommandx{\mkgatebegin}[5][2={},3=ogate,4=.5]{
  \coordinate (gatecord) at (0,0);
  \foreach \n [count=\i] in {#1}{
    \pgfgetlastxy \xc \yc;
    \path (\n);
    \pgfgetlastxy \xn \yn;
    \coordinate (gatecord) at ($(gatecord) + (\xn,0)$);
    \coordinate (gatecord) at ($1/\i*(gatecord)$);
    \ifdim \yn < \yc
    \node (max) at (0,\yc) {};
    \else
    \node (max) at (0,\yn) {};
    \fi
  }
  \coordinate (gatecord) at ($(gatecord) + (0,#4) + (max)$);
  \node[#3,label={below:$#2$}] (#5) at (gatecord) {};
  \pgfgetlastxy{\xgate}{\ygate};
  \pgfmathtruncatemacro{\xgateround}{\xgate};
  \StrCount{#1,}{,}[\l] % from package xxstring
  \ifnum \l < 2 {\errmessage{#1 argument should be a comma-separated list of lenght >= 2}}
  \else{
    \foreach \n in {#1}{
      \path (\n);
      \pgfgetlastxy{\xnode}{\ynode};
      \pgfmathtruncatemacro{\xnround}{\xnode};
      \pgfmathsetmacro\tmpdiff{abs(\xnround - \xgateround)}
      \ifdim \tmpdiff pt > 1 pt \path[line] (#5) -| (\n);
      \else
        \path[line] (#5) -- (\n);
      \fi
    }
  }
  \fi
}

\newcommandx{\mkmerge}[4][2=gatenode,3=i,4=0,usedefault=@]{\mkgateend{#1}[{\ifempty{#3}{}{\nmerge[#3]}}][ogate][#4]{#2}}

\newcommandx{\mkjoin}[4][2=gatenode,3=i,4=0,usedefault=@]{\mkgateend{#1}[{\ifempty{#3}{}{\nmerge[#3]}}][agate][#4]{#2}}

\newcommandx{\mkgateend}[5][2={},3=ogate,4=.5]{
  \coordinate (gatecord) at (0,0);
  \foreach \n [count=\i] in {#1}{
    \pgfgetlastxy \xc \yc;
    \path (\n);
    \pgfgetlastxy \xn \yn;
    \coordinate (gatecord) at ($(gatecord) + (\xn,0)$);
    \coordinate (gatecord) at ($1/\i*(gatecord)$);
    \ifdim \yn > \yc
    \node (min) at (0,\yc) {};
    \else
    \node (min) at (0,\yn) {};
    \fi
  }
  \coordinate (gatecord) at ($(gatecord) - (0,#4) + (min)$);
  \node[#3,label={above:$#2$}] (#5) at (gatecord) {};
  \pgfgetlastxy{\xgate}{\ygate};
  \pgfmathtruncatemacro{\xgateround}{\xgate};
  \StrCount{#1,}{,}[\l] % from package xxstring
  \ifnum \l < 2 {\errmessage{#1 argument should be a comma-separated list of lenght >= 2}}
  \else{
    \foreach \n in {#1}{
      \path (\n);
      \pgfgetlastxy{\xnode}{\ynode};
      \pgfmathtruncatemacro{\xnround}{\xnode};
      \pgfmathsetmacro\tmpdiff{abs(\xnround - \xgateround)}
      \ifdim \tmpdiff pt > 1 pt \path[line] (\n) |- (#5);
      \else
        \path[line] (\n) -- (#5);
      \fi
    }
  }
  \fi
}

\newcommand{\gatedistancein}{3pt}
\newcommand{\gatedistanceinand}{2pt}

\usetikzlibrary{
  arrows,snakes,shapes,automata,backgrounds,petri,positioning,fit,calc,
  decorations.markings,decorations.text,shadows,fadings,patterns,
  decorations.pathreplacing,chains,spy
}

\tikzset{
  src/.style={draw,circle,fill=white,
    minimum size=2mm,
    inner sep=0pt
  },
  sink/.style={draw,circle,double,fill=white,
    minimum size=1.5mm,
    inner sep=0pt
  },
  node/.style={draw,circle,fill=black,
    minimum size=2mm,
    inner sep=0pt
  },
  source/.style={draw,circle,fill=white,
    minimum size=3mm,
    inner sep=0pt
  },
  sink/.style={draw,circle,double,fill=white,
    minimum size=3mm,
    inner sep=0pt
  },
  block/.style = {rectangle, draw=gray, align=center, fill=orange!25, rounded corners=0.1cm,
    minimum size=5mm, inner sep=2pt},
  prenode/.style = {minimum size=9pt,inner sep=2pt, font=\Large},
  bblock/.style = {rectangle, draw=blue!50, opacity=.5, line width=1pt, align=center, fill=white, rounded corners=0.1cm,
    minimum size=7mm, inner sep=2pt},
  prenode/.style = {minimum size=9pt,inner sep=2pt, font=\Large},
  agate/.style={draw, rectangle,
    minimum size=3mm,
    inner sep=0pt,
    fill=orange!25,
    postaction={path picture={% 
        \draw[red]
        ([yshift=\gatedistanceinand]path picture bounding box.south) --
        ([yshift=-\gatedistanceinand]path picture bounding box.north) ;}}
  },
  ogate/.style = {
    diamond, draw, fill=orange!25,
    minimum size=4mm,
    inner sep=0pt,
    postaction={path picture={% 
        \draw[red]
        ([yshift=\gatedistancein]path picture bounding box.south) -- ([yshift=-\gatedistancein]path picture bounding box.north)
        ([xshift=-\gatedistancein]path picture bounding box.east) -- ([xshift=\gatedistancein]path picture bounding box.west)
        ;}}},
  altogate/.style = {
    diamond, draw,
    minimum size=4mm,
    inner sep=0pt,
    postaction={path picture={% 
        \draw
        ([yshift=\gatedistancein]path picture bounding box.south) -- ([yshift=-\gatedistancein]path picture bounding box.north)
        ([xshift=-\gatedistancein]path picture bounding box.east) -- ([xshift=\gatedistancein]path picture bounding box.west)
        ;}}},
  altgate/.style={draw, rectangle,
    minimum size=3mm,
    inner sep=0pt,
    postaction={path picture={% 
        \draw
        ([yshift=\gatedistanceinand]path picture bounding box.south) --
        ([yshift=-\gatedistanceinand]path picture bounding box.north) ;}}},
  anygate/.style = {circle, draw, fill=white,
    minimum size=4mm,
    inner sep=0pt,
    postaction={path picture={% 
        \draw[black]
        ([xshift=-\gatedistancein,yshift=\gatedistancein]path picture bounding box.south east) --
        ([xshift=\gatedistancein,yshift=-\gatedistancein]path picture bounding box.north west)
        ([xshift=-\gatedistancein,yshift=-\gatedistancein]path picture bounding box.north east) --
        ([xshift=\gatedistancein,yshift=\gatedistancein]path picture bounding box.south west)
        ;}}
  },
  smallglobal/.style={
        node distance=1cm and 0.8cm, semithick, scale=0.8, every node/.style={transform shape}
  },
  elli/.style = {draw,densely dotted,-},
  line/.style = {draw,->, rounded corners=0.07cm,>=latex},
  nline/.style = {draw,semithick, ->},
  pline/.style = {draw,->,>=latex},
  node distance=1cm and 0.7cm,
  baseline=(current  bounding  box.center),
  local/.style={rectangle, draw, fill=\fillcolor, drop shadow,
    text centered, rounded corners, minimum height=5em
  },
  bigar/.style={
    draw,very thick, ->
  },
  process/.style={rectangle, draw=gray, fill=\fillcolor, drop shadow,
    text centered, minimum height=5em,text=gray
  },
  choreo/.style={rectangle, draw, fill=\fillcolor, drop shadow,
    text centered, rounded corners, minimum height=5em
  },
  mycfsm/.style={
        font=\footnotesize,
        initial where=above,
        ->,>=stealth,auto, node distance=1cm and 1cm,
        scale=1, every node/.style={transform shape},
        every state/.style=inner sep=2pt,
        baseline=(current  bounding  box.center)
  },
  machinecloud/.style={
    cloud, cloud puffs=10, cloud ignores aspect, minimum height=.1cm, minimum width=2cm, draw
  },
  fitting node/.style={
    inner sep=0pt,
    fill=none,
    draw=none,
    reset transform,
    fit={(\pgf@pathminx,\pgf@pathminy) (\pgf@pathmaxx,\pgf@pathmaxy)}
  },
  mypetri/.style={
    font=\footnotesize,
    baseline=(current  bounding  box.center)
  },
  silentrans/.style = {rectangle, draw=black, align=center, fill=black,
    minimum height=1pt,
    minimum width=15pt,
    inner sep=1.5pt
  },
  reset transform/.code={\pgftransformreset},
  tmtape/.style={draw,minimum size=1.2cm}
}

\newcommand{\gunlessop}{\mbox{\colorOp\tiny\tt unless}}

\newcommandx{\gtry}[5][1=\gname,2={\aG_1 \gchoop \cdots \gchoop \aG_n},3=\gin,4=\gout,5={j},usedefault=@]{
  \def\foo{\gtryop\ {#2} \ \gcatchop\ {#3} {\colorOp \Rightarrow} {#4} {\colorOp \bullet} {\gname[{#5}]}}
  \gnode[{#1}][{\ifempty{#1} {\foo } { \big(\foo \big) }}]
}

\newcommandx{\gtrycatch}[4][1=\gname,2={\aG},3=\gin,4={\aG'},usedefault=@]{
  \def\foo{\gtryop\ {#2} \ \gcatchop\ {#3} \gdoop\ {#4}}
  \gnode[{#1}][{\ifempty{#1} {\foo} {\big( \foo \big) }}]
}

\def\colorGuard{\color{cyan}}
\newcommand{\aguard}{{\colorGuard \phi}}
\newcommandx{\agG}[2][1={\aG},2=\aguard]{{#1} \ifempty{#2}{}{\ \gunlessop\ {#2}}}

\newcommandx{\grcho}[4][1=\gname,2={\agG},3={\agG[\aG'][\aguard']},4={\cdots},usedefault=@]{
  \def\foo{{#2} {\ \ifempty{#4}{\gchoop}{\gchoop \ {#4}\  \gchoop}\ } {#3}}
  \gnode[{#1}][\ifempty{#1}{\foo}{\big( \foo \big)}]
}

\newcommandx{\ggprefix}[3][1=\ptp,2={\aR},3={\aR'},usedefault=@]{f_{#1}} % it was \newcommandx{\common}{...}
\newcommand{\aconfigfn}{\chi}
\newcommand{\aconfig}{\ell}

\newcommand{\lstates}{\statemap}
\newcommandx{\sysconfig}[3][1=\lstates,2=\aconfigfn,3={},usedefault=@]{
  \conf{ {#1},{#2} \ifempty{#3}{}{, #3} }
}
\newcommand{\sysctxfn}[1][]{\gamma_{#1}}
\newcommandx{\sysctx}[2][1=\aQ,2={},usedefault=@]{({#1},\sysctxfn[{#2}])}

\newcommandx{\alog}[4][1=\msg,2=q,3=\gname,4=t,usedefault=@]{\big({#1},{#2},{#3},{#4}\big)}

\newcommand{\aCM}{M}\newcommand{\aM}{\aCM}
\newcommand{\aQ}{Q}
\newcommandx{\aQzero}[1][1=,usedefault=@]{
  {\ifempty{#1}{q_0}{q_{0#1}}}
}
\newcommand{\badbranches}[1][]{\beta\ifempty{#1}{}{\big({#1}\big)}}
\newcommand{\aTrs}{\tset}
\newcommandx{\guardedaction}[2][1=\al,2=\aguard,usedefault=@]{
  {#1} \ifempty{#2}{}{/} {#2}
}
\newcommandx{\atrM}[4][1=q,2=\al,3={\hat q,\hat \al, \aguard},4=q',usedefault=@]{
  {#1} \xrightarrow[{#3}]{\guardedaction[{#2}][]} {{#4}}
}
\newcommandx{\atrS}[5][
  1={\sysconfig[@][@][\badbranches]},
  2=\al,
  3=\aguard,
  4={\sysconfig[\lstates'][\aconfigfn'][\badbranches]},
  5=\sysctx,usedefault=@
]{
  {#1} \xRightarrow{\qquad} {{#4}}
}
\newcommandx{\arevtrS}[2][
  1={\sysconfig[@][@][\badbranches]},
  2={\sysconfig[\lstates'][\aconfigfn'][\badbranches']},
  usedefault=@
]{
  {#1} \rightsquigarrow {#2}
}
\newcommand{\aCS}{S}
\newcommand{\abuffer}{b}

\newcommandx{\enables}[2][1=\aconfigfn,2=\aguard,usedefault=@]{{#1} \vdash {#2}}
\newcommandx{\gprojfn}[5][1=\aG,2=\ptp,3=\cinit,4=\cfinal,5={},usedefault=@]{
  \mathbf{proj}_{#2}({#1},{#3},{#4}\ifempty{#5}{}{,{#5}})
}

\newcommandx{\rbp}[3][1=\aG,2=\aconfigfn,3=\achan,usedefault=@]{\mathtt{RBP}_{{#1},{#2}}\ifempty{#3}{}{\big({#3}\big)}}

\newcommand{\apseudoCFSM}{\mathtt{M}}
\newcommandx{\pseudoseq}[2][1=\apseudoCFSM,2=\apseudoCFSM',usedefault=@]{{#1}  ; {#2}}
\newcommandx{\pseudoCFSM}[4][1=\aQ,2=\aQzero,3=\cfinal,4=\aTrs,usedefault=@]{(#1 \ ; #2 \ ; #3 \ ; #4)}
\newcommandx{\markt}[3][1=\hat{\al},2=\hat{q},3=\aguard,usedefault=@]{\%\big({#1} , {#2}, {#3}\big)}

\newcommandx{\borderfn}[2][1=\aconfig,2=\aloop,usedefault=@]{
  \mathsf{border}_{{#2}}\ifempty{#1}{}{\big({#1}\big)}
}

%% file: intro.tex
\emph{Asynchronous message-passing} is a widely adopted paradigm for
the specification, design, and implementation of communication-centred
applications or systems.
This paradigm has been used at different abstraction levels, including
formal models (e.g. $\pi$-calculus~\cite{sd01,mil99} and communicating
automata~\cite{bz83}), specification languages
(e.g. \emph{message-sequence charts} (MSCs)~\cite{itu11}),
choreography languages (e.g. global calculus~\cite{chy07} and
WS-CDL~\cite{wscdl}), programming languages (e.g. actor models for
Erlang, Scala, and Go).

Choreographic approaches are gaining momentum to handle the complexity
of distributed systems~\cite{kum17}.
These frameworks envisage two views: a global specification and a
local one.
The former defines the order and constraints under which messages are
sent and received, while the local view defines the behavior of each
participant.
The composition of local participants should respect the global
specification.
In this setting, the \emph{realisability} of the global specifications
becomes a concern since there could be some specifications that are
impossible to implement using the local views in a given communication
model.

We propose a general semantic representation based on partially
ordered multisets (pomsets)~\cite{pratt1986modeling}, capable
of specifying global behaviors and analyze their realisability in
terms of asynchronous message-passing.
Our framework assumes asynchronous point-to-point communications and
features a notion of realisability that
\begin{enumerate}
\item \label{it:term} rules out systems where some participants cannot
  ascertain termination
\item \label{it:par} admits multi-threaded participants
\item \label{it:syn} allows us to define syntax-oblivious conditions
\item \label{it:decide} can be decided by an analysis of the partial
  orders of communication events.
\end{enumerate}
These features have several practical advantages.
Indeed, by~\eqref{it:term}, we admit systems where participants may
get stuck on some messages, only if that is specified in the global
model.
The use of multi-threaded participants~\eqref{it:par} makes our
framework more expressive than existing ones (see discussion on
this point in~\cite{gt18}).
Syntax independent conditions~\eqref{it:syn} are applicable to
different global models.
Finally, \eqref{it:decide} enables the identification of design errors
in global models rather than in execution traces where they are harder
to analyse.

\paragraph{Outline}
\cref{sec:msc} gives the basic definitions.  \cref{sec:realisability}
introduces the problems of realisability and sound termination; also,
it provides verification conditions in the style of~\cite{aey03}.
\cref{sec:pomsets} presents the sufficient conditions for
realisability and sound termination that can be tested over partial
orders.
\cref{sec:implementation} discusses the complexity of the new verification
conditions.
Finally, \cref{sec:related} discusses related work and
\cref{sec:conc} draws some conclusions.

%%% Local Variables:
%%% mode: latex
%%% TeX-master: "main"
%%% End:

%% file: msc.tex
We collect the main definitions needed in the rest of the paper.
The material of this section is not an original
contribution\footnote{Except for the different definition of accepting
  states of
  communicating automata.}
and it is presented only to make the paper self-contained borrowing
and combining definitions and notations
from~\cite{gaifman1987partial,aey03,katoen1998pomsets,bz83}.

We borrow the formalisation of partially-ordered multi-set
of~\cite{gaifman1987partial}.
\begin{definition}[Lposets]
  A labelled partially-ordered set (lposet) is a triple
  $(\eset, \leq, \alf)$, with $\eset$ a set of events,
  $\leq \subseteq \eset \times \eset$ a reflexive, anti-symmetric, and
  transitive relation on $\eset$, and $\alf: \eset \rightarrow \lset$
  a labelling function mapping events in $\eset$ to labels in $\lset$.
\end{definition}
Intuitively, $\leq$ represents causality; for $\ae \neq \ae'$, if
$\ae \leq \ae'$ and both events occur then $\ae'$ is caused by $\ae$.
Note that $\alf$ is not required to be injective: for
$\ae \neq \ae' \in \eset$, $\alf(\ae) = \alf(\ae')$ means that $\ae$
and $\ae'$ model different occurrences of the same action.

\begin{definition}[Pomsets]\label{def:pomsets}
  Two lposets $(\eset, \leq, \alf)$ and $(\eset', \leq', \alf')$ are
  \emph{isomorphic} if there is a bijection
  $\phi: \eset \rightarrow \eset'$ such that $\ae \leq \ae' \iff
  \phi(\ae) \leq' \phi(\ae')$ and $\alf = \alf' \circ \phi$.
  A partially-ordered multi-set (of actions), pomset for short, is an
  isomorphism class of lposets.
\end{definition}
Using pomsets in place of lposets allows us to
abstract away from the names of events in
$\eset$.
In the following, $[\eset, \leq, \alf]$ denotes the isomorphism class
of $(\eset, \leq, \alf)$, symbols $\apom,\apom', \dots$ (resp.
$\aR, \aR', \dots$) range over (resp. sets of) pomsets, and we assume
that any $\apom$ contains at least one lposet which will possibly be
referred to as $(\esetof \apom$, $\leqof \apom, \alfof \apom)$.
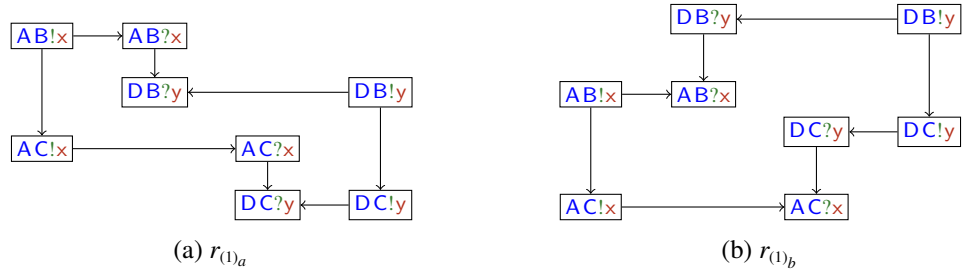
\begin{figure*}[t!]
  \centering
  \begin{subfigure}[b]{0.45\textwidth}
    \centering
    \begin{tikzpicture}[every node/.style = {rectangle,draw}, transform shape]\tiny
      \node (ab!x) at (0,0)         {$\aout[@][@][][{\msg[x]}]$};
      \node (ac!x) at (0,-1.5)      {$\aout[@][C][][{\msg[x]}]$};
      \node (ab?x) at (1.5,0)       {$\ain[@][@][][{\msg[x]}]$};
      \node (db?y) at (1.5,-0.75)   {$\ain[D][@][][{\msg[y]}]$};
      \node (ac?x) at (3,-1.5)      {$\ain[@][C][][{\msg[x]}]$};
      \node (dc?y) at (3,-2.25)     {$\ain[D][C][][{\msg[y]}]$};
      \node (db!y) at (4.5,-0.75)   {$\aout[D][B][][{\msg[y]}]$};
      \node (dc!y) at (4.5,-2.25)   {$\aout[D][C][][{\msg[y]}]$};
      \path[->,draw] (ab!x) -- (ab?x);
      \path[->,draw] (ac!x) -- (ac?x);
      \path[->,draw] (db!y) -- (db?y);
      \path[->,draw] (dc!y) -- (dc?y);
      \path[->,draw] (ab!x) -- (ac!x);
      \path[->,draw] (db!y) -- (dc!y);
      \path[->,draw] (ab?x) -- (db?y);
      \path[->,draw] (ac?x) -- (dc?y);
    \end{tikzpicture}
    \caption{$\apom_{\eqref{fig:example:msg}_a}$}
  \end{subfigure}
  \begin{subfigure}[b]{0.45\textwidth}
    \centering 
    \begin{tikzpicture}[every node/.style = {rectangle,draw}, transform shape]\tiny
      \node (ab!x) at (0,0)    {$\aout[@][@][][{\msg[x]}]$};
      \node (ac!x) at (0,-1.5) {$\aout[@][C][][{\msg[x]}]$};
      \node (db?y) at (1.5,1)    {$\ain[D][@][][{\msg[y]}]$};
      \node (ab?x) at (1.5,0)    {$\ain[@][@][][{\msg[x]}]$};
      \node (dc?y) at (3,-.5)  {$\ain[D][C][][{\msg[y]}]$};
      \node (ac?x) at (3,-1.5) {$\ain[@][C][][{\msg[x]}]$};
      \node (db!y) at (4.5,1)    {$\aout[D][B][][{\msg[y]}]$};
      \node (dc!y) at (4.5,-.5)  {$\aout[D][C][][{\msg[y]}]$};
      \path[->,draw] (ab!x) -- (ab?x);
      \path[->,draw] (ac!x) -- (ac?x);
      \path[->,draw] (db!y) -- (db?y);
      \path[->,draw] (dc!y) -- (dc?y);
      \path[->,draw] (ab!x) -- (ac!x);
      \path[->,draw] (db!y) -- (dc!y);
      \path[->,draw] (db?y) -- (ab?x);
      \path[->,draw] (dc?y) -- (ac?x);
    \end{tikzpicture}
    \caption{$\apom_{\eqref{fig:example:msg}_b}$}
  \end{subfigure}
  \caption{\label{fig:example:msg} $\aR_{\ref{fig:example:msg}}$ =
    $\{\apom_{\eqref{fig:example:msg}_a},
    \apom_{\eqref{fig:example:msg}_b}\}$ is a set of two pomsets Two
    pomsets
  }
\end{figure*}
An event $\ae$ is an \emph{immediate predecessor} of an event $\ae'$
in a pomset $\apom$ if $\ae \neq \ae'$, $\ae \leqof \apom \ae'$, and
for all $\ae'' \in \esetof \apom$ such that
$\ae \leqof \apom \ae'' \leqof \apom \ae'$ either $\ae = \ae''$ or
$\ae' = \ae''$.
If $\ae$ is an immediate predecessor of $\ae'$ in $\apom$ then $\ae'$
is an \emph{immediate successor} of $\ae$ in $\apom$.

Hereafter, we consider pomsets labelled by communications representing
output and input actions between a sender and a receiver.
Technically, this is done by instantiating the set $\lset$ of labels
as follows.

Let $\ptpset$ be a set of \emph{participants} (ranged over by $\p$,
$\q$, etc.), $\msgset$ a set (of types) of \emph{messages} (ranged
over by $\msg$, $\msg[x]$, etc.).
We take $\ptpset$ and $\msgset$ disjoint.
Participants coordinate with each other by exchanging messages over
\emph{communication channels}, that are elements of the set
$\chset = (\ptpset \times \ptpset) \setminus \{(\p,\p) \sst \p \in
\ptpset\}$ and we abbreviate $(\p,\q) \in \chset$ as $\achan$.
The set of \emph{(communication) labels} $\lset$ is defined by
\[
\lset = \lset^! \cup \lset^?
\qquad\text{where}\qquad
\lset^! = \chset \times  \{!\} \times \msgset
\qquad\text{and}\qquad
\lset^? = \chset \times   \{?\} \times \msgset
\]
The elements of $\lset^!$ and $\lset^?$, outputs and inputs,
respectively represent \emph{sending} and \emph{receiving} actions; we
shorten $(\achan,!,\msg)$ as $\aout$ 
and $(\achan,?,\msg)$ as $\ain$
and let $\al$, $\al'$, $\ldots$ range over $\lset$.
The \emph{subject} of an action is defined by
\begin{align*}
  \esubject[\aout] = \p
  \quad \text{(the sender)}
  \qquad\text{and}\qquad
  \esubject[\ain] = \q
  \quad
  \text{(the receiver)}
\end{align*}
We will represent pomsets as the (variant\footnote{Edges of Hasse
  diagrams are usually not oriented; here we use arrow so to draw
  order relations between events also horizontally.} of) Hasse diagram
of the immediate predecessor relation as done in the examples of
\cref{fig:example:msg}.
For instance, in the pomset $\apom_{\eqref{fig:example:msg}_a}$ the
input event of $\q$ from $\p$ immediately precedes the input of $\q$
from $\ptp[d]$ while the events with those labels are in the reversed
order in $\apom_{\eqref{fig:example:msg}_b}$.

\begin{definition}[Projection of pomsets]
  The \emph{projection $\projpom \apom \p$ of a pomset $\apom$ on a participant
    $\p \in \ptpset$} is obtained by restricting $\apom$ to the events
  having subject $\p$: formally
  $\projpom \apom \p = [\eset_{\apom, \p},\ \leqof{\apom} \cap\
  (\eset_{\apom, \p} \times \eset_{\apom, \p}), \ \alfof{\apom}
  \fnrestriction{\eset_{\apom,\p}}]$ where
  $\eset_{\apom, \p} = \{\ae \in \esetof \apom \sst \esubject[\alfof
  \apom(\ae)] = \p \}$.
\end{definition}

Pomsets are a quite expressive model of global views of
choreographies~\cite{gt18}; in fact, MSCs\footnote{Pomsets can also be
  used to give semantics to the composition of MSCs;
  see~\cite{katoen1998pomsets}.
} can be defined as a subclass of pomsets.
\begin{definition}[Well-formedness, completeness, and MSCs]\label{def:msc}
  A pomset $\apom$ over $\lset$ is \emph{well-formed} if for every
  event $\ae \in \esetof{\apom}$
  \begin{enumerate}
  \item \label{it:out} if $\alfof{\apom}(\ae) = \aout$, there is at
    most one $\ae' \in \esetof \apom$ immediate successor of $\ae$
    in $\apom$
    with $\alfof{\apom}(\ae') = \ain$ (and, if such $\ae'$ exists, we
    say that $\ae$ and $\ae'$ \emph{match} each other)
  \item \label{it:in} if $\alfof{\apom}(\ae) = \ain$, there exists
    exactly one $\ae' \in \esetof{\apom}$ immediate predecessor of
    $\ae$ in $\apom$ with $\alfof{\apom}(\ae') = \aout$
  \item \label{it:loc} for each $\ae' \in \esetof \apom$, if $\ae$
    is an immediate predecessor of $\ae'$ and
    $\esubject[{\alfof \apom (\ae)}] \neq \esubject[{\alfof \apom(\ae')}]$ then
    $\ae$ and $\ae'$ are matching output and input events respectively
  \item \label{it:deg} for each $\ae' \neq \ae \in \esetof \apom$ with
    $\alfof{\apom}(\ae) = \alfof{\apom}(\ae') = \aout$, and for all
    $\bar \ae , \bar \ae' \in \esetof \apom$ immediate successors in
    $\apom$ of $\ae$ and of $\ae'$ respectively if
    $\alfof{\apom}(\bar \ae) = \alfof{\apom}(\bar \ae') = \ain$ and
    $\ae \leqof{\apom} \ae'$ then
    $\bar \ae' \not\leqof{\apom} \bar \ae$
  \end{enumerate}
  All conditions of \cref{def:msc} are straightforward but the last one,
  which requires that ordered output events with the same label cannot
  be matched by inputs that have opposite order.
  Pomset $\apom$ is \emph{complete} if there is no send event in
  $\esetof \apom$ without a matching receive event.

  A \emph{message-sequence chart} is a well-formed and complete pomset
  $\apom$ such that
  $\leqof{\projpom{\apom}{\p}}$ is a total order, for every $\p \in
  \ptpset$.
\end{definition}

Well-formed pomsets permit to represent inter-participant concurrency
since they keep independent not matching communication events of
different participants.
Also, well-formed pomsets allow intra-participant concurrency (i.e.
multi-threaded participants) since they do not require
$\leqof{\projpom{\apom}{\p}}$ to be totally ordered. MSCs
are obtained by restricting participants to be single-threaded.
The pomsets in \cref{fig:example:msg} are indeed MSCs describing
different orders of the same set of events.
Vertical arrows represent orders on the events of a participant; for
instance, the leftmost vertical arrow of
$\apom_{\eqref{fig:example:msg}_a}$ represents that the output of $\p$
to $\q$ precedes the one to $\ptp[c]$.
Basically, vertical arrows correspond to the \emph{projections} of the
pomsets on participants; these projections are obtained by restricting
$\apom_{\eqref{fig:example:msg}_a}$ and
$\apom_{\eqref{fig:example:msg}_b}$ to the events having the same
subject.
More precisely, the projection on one of the participants consists of
the $i$-th vertical arrow where $i$ is the alphabetical order of the
participant (e.g., the projection of $\ptp[C]$ is the third arrow).
The behaviour of $\p$ (and $\ptp[D]$) is the same in both MSCs: $\p$
(resp. $\ptp[D]$) first sends message $\msg[x]$ (resp.  $\msg[y]$)
to $\q$ and then to $\ptp[C]$.
The behaviour of $\q$ (and $\ptp[C]$) differs: in
$\apom_{\eqref{fig:example:msg}_a}$, $\q$ first receives the message
from $\p$ then the one from $\ptp[D]$, in
$\apom_{\eqref{fig:example:msg}_b}$, $\q$ has the same interactions
but in opposite order.
Likewise for $\ptp[C]$.

Well-formed pomsets capture the semantics of choreographic modelling languages; we used them to give semantics of choreographies in~\cite{gt18}.
In particular, 
to handle distributed choices of choreographies one uses sets of pomsets $\aR$, so that
each $\apom \in \aR$ yields the causal dependencies of the communications in a branch.
For instance, the set
$\aR_{\eqref{fig:example:msg}} = \{\apom_{\eqref{fig:example:msg}_a},
\apom_{\eqref{fig:example:msg}_b}\}$
represents a choice between the fact that $\q$ may receive messages
$\msg[x]$ and $\msg[y]$ in any order.

A natural question to ask is:
\begin{quote}
  \quo{is it possible to realise $\aR_{\eqref{fig:example:msg}}$ with
    asynchronously communicating local views?}
\end{quote}
The next section answers this question for pomsets similarly to what
done in~\cite{aey03} where closure conditions for MSCs where
identified.
%
 
%%% Local Variables:
%%% mode: latex
%%% TeX-master: "main"
%%% End:

%% file: realisability.tex
Hereafter we assume all structures, including languages, words
 and pomsets, to be finite. 
Given a pomset $\apom$, a \emph{linearization} of $\apom$ is a string
in $\lset^\star$ obtained by considering a total ordering of the
events $\esetof{\apom}$ that is consistent with the partial order
$\leqof{\apom}$ , and then replacing each event by its label.
More precisely, let $\size{\esetof{\apom}}$ be the cardinality of
$\esetof \apom$, a word
$\aW = \alfof{\apom}(\ae_1) \dots
\alfof{\apom}(\ae_{\size{\esetof{\apom}}})$ is a \emph{linearization}
of a pomset $\apom$ if 
$\ae_1 \dots \ae_{\size{\esetof{\apom}}}$ 
 is a permutation that totally orders the
events in $\esetof{\apom}$ so that if $\ae_i \leqof{\apom} \ae_j$ then
$i \leq j$.
For a pomset $\apom$, define $\rlang(\apom)$ to be the set of all
linearizations of $\apom$.
A word $\aW$ over $\lset$ is \emph{well-formed}
(resp. \emph{complete}) if it is the linearization of a well-formed
(resp. \emph{complete}) pomset.
Hereafter, for a word $\aW \in \lset^\star$, $\projword{\aW}{\p}$
denotes the projection of $\aW$ that retains only those events where
participant $\p \in \ptpset$ is the subject.
Operation $\projword \_ \p$ acts element-wise on languages over
$\lset$.
The \emph{language} of a set of pomsets $\aR$ is simply defined as
$\rlang(\aR) = \bigcup_{\apom \in \aR} \rlang(\apom)$.

Local views are often conveniently modelled in terms of
\emph{communicating automata} of some sort.
An $\p$-\emph{communicating finite state machine ($\p$-CFSM)}
$\aCM = (\aQ,\aQzero,\aQfinal,\aTrs)$ is a finite-state automaton on
the alphabet $\lset$ such that, $\aQzero \in \aQ$ is the initial state,
$\aQfinal \subseteq \aQ$ are the accepting states, and
for each $q \trans{\al} {q'}$ holds
$\esubject[\al] = \p$.
A \emph{(communicating) system} is a map
$\aCS = (\aCM_{\p})_{\p \in \ptpset}$ assigning an $\p$-CFSM
$\aCM_{\p}$ to each participant $\p \in \ptpset$.
For all $\p \neq \q \in \ptpset$, 
we shall use an unbounded multiset
 $\abuffer_{\achan}$ where $\aCM_{\p}$ puts the message to
$\aCM_{\q}$ and from which $\aCM_{\q}$ consumes the messages from
$\aCM_{\p}$.

The semantics of communicating systems is defined in terms of 
transition relations between \emph{configurations} which keep track of
the state of each machine and the content of each buffer.
Let $\aCS = (\aCM_{\p})_{\p \in \ptpset}$ be a communicating
system.
A \emph{configuration} of $\aCS$ is a pair
$\aConf = \csconf q \abuffer$ where
$\vec q = (q_{\p})_{\p \in \ptpset}$ maps each participant $\p$ to its
local state $q_{\p} \in \aQ_{\p}$ and
$\vec \abuffer = (\abuffer_{\achan})_{\achan \in\chset}$ where
the buffer
$\abuffer_{\achan} : \msgset \to \mathbb{N}$ is a map assigning the
number of occurrences of each message; state $q_{\p}$
keeps track of the state of the automaton $\aCM_{\p}$ and buffer
$\abuffer_{\achan}$ keeps track of the messages sent from $\p$ to
$\q$.
The \emph{initial} configuration $\aConf_0$ is the one where, for all
$\p \in \ptpset$, $q_{\p}$ is the initial state of the corresponding
CFSM and all buffers are empty.
Given two configurations $\aConf = \csconf q \abuffer$ and
$\aConf' = \csconf{q'}{\abuffer'}$, relation
$\aConf \TRANSS{\ \al \ } \aConf'$ holds if there is a message
$\msg \in \msgset$ such that either (1) or (2) below holds:
\begin{center}
  \begin{tabular}{lr}
    \begin{minipage}{.4\linewidth}\small
      1. 
        $\al = \aout$ and $q_{\p} \trans[\p]{\al} {q'_{\p}}$
        and 
        \begin{itemize}
        \item[a.] $q'_{\p[C]} = q_{\p[C]}$ for all ${\p[C]} \neq \p \in \ptpset$ and
        \item[b.] $\abuffer'_{\achan} = \upd{\abuffer_{\achan}}{\msg}{\abuffer_{\achan}(\msg)\ +\ 1}$\\
        \end{itemize}
    \end{minipage}
    &
    \begin{minipage}{.5\linewidth}\small
      2.
        $\al = \ain[@][@][@]$ and $q_{\q} \trans[{\q}]{\al} {q'_{\q}}$
        and 
        \begin{itemize}
        \item[a.] $q_{\p[C]}' = q_{\p[C]}$ for all ${\p[C]} \neq \q \in \ptpset$  and
        \item[b.] $\abuffer_{\achan}(\msg) > 0$ and $\abuffer'_{\achan} = \upd{\abuffer_{\achan}}{\msg}{\abuffer_{\achan}(\msg) - 1}$
          \\
        \end{itemize}
    \end{minipage}
  \end{tabular}
\end{center}
where, $\upd f x y$ is the usual notation for the updating of a
function $f$ in a point $x$ of its domain with a value $y$.
Condition (1) puts $\msg$ on channel $\achan$, while (2) gets
$\msg$ from channel $\achan$ by simply updating the number of occurrences
of $\msg$ in the buffer $\abuffer_{\achan}$.
In both cases, any machine or buffer not involved in the transition
is left unchanged in the new configuration $\aConf'$.

The automata model adopted in~\cite{aey03} is a slight variant of
\emph{communicating-finite state machines} (CFSMs)~\cite{bz83}.
The two models have the same definition of automata; they differ in
how communication is attained, but are equivalent up to internal
transitions (which in~\cite{aey03} have been used to simplify proofs).
We used the definition of CFSMs in~\cite{bz83} to encompass accepting
states (necessary to define our notion of 
termination soundness~\cref{def:termination}).
Another minor deviation from the definition of CFSMs introduced
in~\cite{bz83} is that buffers become multisets in~\cite{aey03} while
in~\cite{bz83} they follow a FIFO policy.

Given a communicating system $\aCS$, a configuration
$\aConf= \csconf {q} {\abuffer}$ of $\aCS$ is ($i$) \emph{accepting}
if all buffers in $\vec \abuffer$ are empty and the local state $\vec q(\p)$
of each participant $\p$ is accepting while ($ii$) $\aConf$ is a
\emph{deadlock} if no accepting configuration is reachable from
$\aConf$.
We can then define the \emph{language of $\aCS$} as the set
$\langof{\aCS} \in \lset^{\star}$ of sequences $\al_0 \dots \al_{n-1}$
such that
$\aConf_0 \TRANSS{\al_0} \dots \TRANSS{\al_{n-1}} \aConf_{n}$ and
$\aConf_{n}$ is an accepting configuration.

% ******************************
% DEFINITIONS
% ******************************
The notion of \emph{realisability} and \emph{sound
  termination} (cf. \cref{def:langreal,def:termination} below) are
given in terms of the relation between the \emph{language} of the
global view and the one of a system of local views \quo{implementing}
it. 
Our notion of realisability considers languages over $\lset$ as sets
of traces of the distributed executions of some CFSMs, analogously to
\cite{aey03}.
\begin{definition}[Realisability]\label{def:langreal}
  A language $\rlang \subseteq \lset^\star$ is \emph{weakly
    realisable} if there is a communicating system $\aCS$ such that
  $\rlang = \langof{\aCS}$; when $\aCS$ is deadlock-free we say that
  $\rlang$ is \emph{safely realisable}.
  A set of pomsets $\aR$ is \emph{weakly \emph{(resp. \emph{safely})}
    realisable} if $\langof{\aR}$ is weakly (resp. safely) realisable.
\end{definition}

The notion of realisability is meaningful when pomsets are
\emph{well-formed} and \emph{complete}, namely when they yield a
proper match among receive and send events.

In general, safe realisability is not enough to rule out undesirable
designs.
In fact, it admits systems where participants cannot ascertain
termination and may be left waiting forever for some messages.
This may lead non-terminating participants to unnecessarily lock
resources once the coordination is completed.
We explain this considering \cref{fig:example:term1} which can be
interpreted as follows.
Participant $\p$ starts a transaction with $\q$ by sending message
$\msg[x]$.
Pomset $\apom_{\ref{fig:example:term1}_a}$ represents a scenario where
the transaction was started but neither committed nor aborted.
Pomset $\apom_{\ref{fig:example:term1}_b}$ represents a scenario
where the transaction started and eventually committed.
Yet, $\q$ is uncertain whether message $\msg[y]$ is going to be sent
or not and hence $\q$ be could locally decide to terminate immediately
after receiving $\msg[x]$ leaving $\p[c]$ waiting for message
$\msg[z]$.
However, depending on the application requirements, it may be the case
that termination awareness is important for $\q$ and not for $\p[c]$
because e.g., either $\p[c]$ is not \quo{wasting} resources or it is
immaterial that such resources are left locked.
To handle this limitation we introduce a novel termination condition,
which allows to specify the subset of participants that should be able
to identify when no further message can be exchanged.

\begin{definition}[Termination soundness]\label{def:termination}
  A participant $\p \in \ptpset$ is \emph{termination-unaware} in a
  system $\aCS$ if there exists an accepting configuration
  $\csconf {q} {\abuffer}$ reachable in $\aCS$ having a transition
  departing from $\vec q(\p)$ that is labelled in $\lset^{?}$.
  
  A set of participants $\ptpset' \subseteq \ptpset$
  is
  \emph{termination-aware} in a system $\aCS$ if there is 
  no $\p \in \ptpset$ that is termination-unaware in $\aCS$.
  A language $\rlang$ over $\lset$ is \emph{termination-sound} for
  $\ptpset' \subseteq \ptpset$ if $\rlang$ is safely realisable by a
  system for which $\ptpset'$ is termination-aware.
  A set of pomsets $\aR$ is \emph{termination-sound} for $\ptpset'$
  if $\langof{\aR}$ is termination-sound for $\ptpset'$.
\end{definition}

% ******************************
% VERIFICATION CONDITIONS
% ******************************
Realisability and termination soundness can be established by
analyzing verification conditions of the language.
In \cite{aey03} two closure conditions are introduced that entail weak
and safe realisability.
A word $\aW$ over $\lset$ is \emph{$\pset$-feasible} for
$\rlang \subseteq \lset^\star$ if
$\forall \p \in \ptpset \qst \exists \aW' \in \rlang \qst
\projword{\aW}{\p} = \projword{\aW'}{\p}$.
In~\cite{aey03}, a language $\rlang$ over the alphabet $\lset$
that enjoys the  following conditions
\begin{align*}
  \rlang \supseteq \{\aW \in \lset^\star \sst \aW \text{ well-formed, complete, and $\pset$-feasible for } \rlang\}
\end{align*}
is said\footnote{We stick with the terminology in \cite{aey03} where
  closure conditions are not given specific names.} to be \CC.
Intuitively, the closure condition \CC\ entails that $\rlang$ is
realisable by the set of participants performing the actions in
$\lset$: if each participant cannot tell apart a trace $\aW$ with one
of its expected executions (i.e., those in $\rlang$) then $\aW$ must
be in $\rlang$ or, in the terminology of~\cite{aey03}, $\aW$ is
\emph{implied}.
Closure condition \CC\ characterises the class of weakly realisable
languages over $\lset$.
\begin{theorem}[\cite{aey03}]
  A language $\rlang$ is weakly realisable if, and only if, $\rlang$ contains only
  well-formed and complete words and satisfies \CC.
\end{theorem}

The language of the set of pomsets $\{\apom_{\eqref{fig:example:msg}_a}, \apom_{\eqref{fig:example:msg}_b}\}$ of
\cref{fig:example:msg} is not closed under \CC.
In fact, the well-formed and complete word 
\begin{align}\label{eq:nocc2}
  \aout[\p][\q][][{\msg[x]}] ; \
  \ain[\p][\q][][{\msg[x]}] ; \
  \aout[D][\q][][{\msg[y]}] ; \
  \ain[D][\q][][{\msg[y]}] ; \
  \aout[D][C][][{\msg[y]}] ; \
  \ain[D][C][][{\msg[y]}] ; \
  \aout[\p][C][][{\msg[x]}] ; \
  \ain[\p][C][][{\msg[x]}]
\end{align}
satisfies the conditions of \CC,
because the
projection of the word \eqref{eq:nocc2} on each participant
equals the projection of a linearization of
$\apom_{\eqref{fig:example:msg}_a}$ or of $\apom_{\eqref{fig:example:msg}_b}$ on the same participant. 
However, \eqref{eq:nocc2}
  is not in
the language
$\rlang(\aR_{\eqref{fig:example:msg}})$, because
 $\ain[\p][C][][{\msg[x]}]$ must precede
  $\ain[D][C][][{\msg[y]}]$ in all the words obtained by the
  linearization of $\apom_{\eqref{fig:example:msg}_a}$, while in those obtained by a
  linearization of $\apom_{\eqref{fig:example:msg}_b}$, 
  $\ain[D][\q][][{\msg[y]}]$ must precede
  $\ain[\p][\q][][{\msg[x]}]$.

The realisability entailed by condition \CC\ is \quo{weak} because it
does not rule out possibly deadlocking systems.
Therefore, an additional closure condition, dubbed \CC[3], has
been
identified in~\cite{loh02,aey03}.
A language $\rlang$ over the alphabet $\lset$ has the closure
condition \CC[3] when
\begin{align*}
  \pref \supseteq \{\aW \in \lset^\star \sst \aW \text{ well-formed and $\pset$-feasible for } \pref \}
\end{align*}
where $\pref$ is the prefix closure of $\rlang$.
Basically, condition \CC[3] states that any (partial) execution that
cannot be told apart by any of the participants is a (partial)
execution in $\rlang$.  And now the following result characterises
safe realisability.
\begin{theorem}[\cite{loh02,aey03}]\label{thm:cc3}
  A language $\rlang$ is safe realisable if, and only if, $\rlang$
  contains only well-formed and complete words and satisfies \CC[2]
  and \CC[3] \footnote{The theorem in~\cite{aey03} describes a
    different condition, \CC[2'], which is easier to implement and is
    equivalent to \CC[2] when in conjunction with \CC[3]}.
\end{theorem}
%

% ******************************
% PROJECTIONS
% ******************************
Once a language $\rlang$ is known to be realisable, we get a system
$\aCS(\rlang) = (\aCM_{\p})_{\p \in \ptpset}$ realising $\rlang$ by
defining, for all $\p \in \ptpset$
\begin{align*}
  \aCM_{\p} =
  (\pref[\projword{\rlang}{\p}],\epsilon,\projword{\rlang}{\p},\apomMtrans{})
\end{align*}
where $\aW \apomMtrans{\al} {\aW. \al}$ if
$\aW.\al \in \pref[\projword{\rlang}{\p}]$.
Then, in~\cite{aey03} the following result is shown.
\begin{theorem}[\cite{aey03}]\label{thm:projection}
  If $\rlang$ is a weakly realisable language then
  $\langof{\aCS(\rlang)} = \rlang$.
  Moreover, if $\rlang$ is safely realisable then
  $\aCS(\rlang)$ is deadlock-free.
\end{theorem}

% ******************************
% CONDITION TERMINATION
% ******************************
\newcommand{\langTerm}[1][\ptpset']{
\mathcal{T}_{#1}
}
\begin{figure*}[t!]
    \centering
    \begin{subfigure}[b]{0.2\textwidth}
        \centering
        \begin{tikzpicture}[scale=1, every node/.style = {rectangle,draw}, transform shape]\tiny
        \node (a1) at (0,0) {$\aout[@][@][][{\msg[x]}]$};
        \node (b1) at (1.5,0) {$\ain[@][@][][{\msg[x]}]$};
        \path[->,draw] (a1) -- (b1);
      \end{tikzpicture}
      \caption{$\apom_{\ref{fig:example:term1}_a}$}
    \end{subfigure}%
    ~
    \begin{subfigure}[b]{0.3\textwidth}
      \centering
      \begin{tikzpicture}[scale=1, every node/.style = {rectangle,draw}, transform shape]\tiny
        \node (a1) at (0,0) {$\aout[@][@][][{\msg[x]}]$};
        \node (a2) at (0,-1) {$\aout[@][@][][{\msg[y]}]$};
        \node (b1) at (1.5,0) {$\ain[@][@][][{\msg[x]}]$};
        \node (b2) at (1.5,-1) {$\ain[@][@][][{\msg[y]}]$};
        \node (b3) at (1.5,-2) {$\aout[B][C][][{\msg[z]}]$};
        \node (c1) at (3,-2) {$\ain[B][C][][{\msg[z]}]$};
        \path[->,draw] (a1) -- (a2);
        \path[->,draw] (a1) -- (b1);
        \path[->,draw] (a2) -- (b2);
        \path[->,draw] (b1) -- (b2);
        \path[->,draw] (b2) -- (b3);
        \path[->,draw] (b3) -- (c1);
      \end{tikzpicture}
      \caption{$\apom_{\ref{fig:example:term1}_b}$}
    \end{subfigure}
    \caption{\label{fig:example:term1}A set of two pomsets that is not
      termination sound for $\q$ or $\ptp[C]$}
\end{figure*}
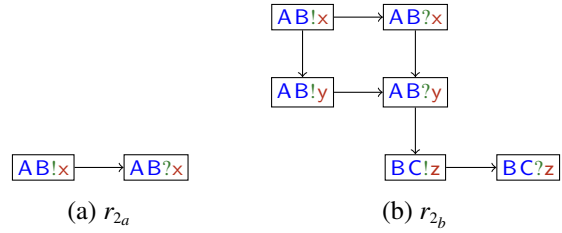

We introduce a new verification condition for
termination soundness.
A participant $\p \in \ptpset$ is \emph{termination-unaware}
for the language  $\rlang$ over $\lset$ 
if there exist $\aW,\aW' \in \rlang$
such that
$\projword{\aW}{\p}$ is a
prefix of $\projword{\aW'}{\p}$ and the first symbol in
$\projword{\aW'}{\p}$ after $\projword{\aW}{\p}$ is in $\lset^{?}$.
Given a set of participants $\ptpset' \subseteq \ptpset$, 
we say that $\rlang$ is $\ptpset'$\emph{-terminating} when there is no
$\p \in \ptpset'$ termination-unaware for $\rlang$.
The language of the family of pomsets
$\{\apom_{\ref{fig:example:term1}_a},
\apom_{\ref{fig:example:term1}_b}\}$ of \cref{fig:example:term1}
is $\{\p\}$-terminating.
However, such language is not $\{\q\}$-terminating.
In fact, after receiving the message $\ain[@][@][][{\msg[x]}]$,
participant $\q$ cannot distinguish whether $\p$ terminates or will
send $\aout[@][@][][{\msg[y]}]$; hence $\q$ ends up in a state where
it is ready to fire the input $\ain[\p][\q][][{\msg[y]}]$, but no
matching output could arrive from $\p$.
And likewise for $\ptp[c]$.

\begin{theorem}
  For $\ptpset' \subseteq \ptpset$,
  if $\rlang$ is $\ptpset'$-terminating and safely realisable
  then it is termination-sound for $\ptpset'$.
\end{theorem}
\begin{proof}
  The proof is trivial.
  Let $\aCS(\rlang)$ be the system
   obtained from the construction of \cref{thm:projection}.  $\aCS(\rlang)$ is deadlock-free and $\rlang = \langof{\aCS(\rlang)}$.
   Let $\p \in \ptpset'$, $\aW \in \rlang$, and $\aConf$ an accepting
   configuration reached in a run of $\aCS$ corresponding to $\aW$.
   For each $\aW' \in \rlang$ such that $\projword{\aW}{\p}$ is prefix
   of $\projword{\aW'}{\p}$, the first symbol in $\projword{\aW'}{\p}$
   after $\projword{\aW}{\p}$ cannot be an input (since $\rlang$ is
   $\ptpset'$-terminating).  Therefore, by construction of
   $\aCS(\rlang)$, there is no input transition departing from the
   local state of $\p$ in $\aConf$.
\end{proof}

%%% Local Variables:
%%% mode: latex
%%% TeX-master: "main"
%%% End:

%% file: pomsets.tex
\newcommand{\lineariz}[1][\p]{\ell_{#1}}

We introduce a different approach to check realisability and sound
termination of specifications, which does not require to explicitly
compute the language of the family of pomsets.  This allows us to
avoid the combinatorial explosion due to interleavings.
The main strategy is to provide alternative definitions of closures
directly on pomsets which handle both \emph{intra-} and
\emph{inter-participant} concurrency.
Besides theoretical benefits, this yields a clear advantage for
practitioners.
In fact, design errors can be identified and confined in more abstract
models, closer to the global specification than to traces of
execution.
Also, our verification conditions require to analyze sets of pomsets; therefore,
they are syntax-oblivious.
As discussed in \cref{sec:implementation}, our conditions strictly entail
the corresponding ones in \cref{sec:realisability}

\newcommand{\normalform}{\square}
\begin{definition}[Closure]\label{def:inter-closure}
  Let $\rho$ be a function from $\ptpset$ to pomsets and
  $(\apom^{\p})_{\p \in \ptpset}$ be the tuple where $r^{\p} = \projpom
  {\rho(\p)} \p$ for all $\p \in \ptpset$.
  The inter-participant closure
  $\normalform((\apom^{\p})_{\p \in \ptpset})$ is the set of all
  well-formed pomsets
  $ \left [ \bigcup_{\p \in \ptpset} \esetof{\apom^{\p}}, \quad
    \leqof{I} \cup \bigcup_{\p \in \ptpset} \leqof{\apom^{\p}}, \quad
    \bigcup_{\p \in \ptpset} \alfof{\apom^{\p}} \right ] $ where
  $\leqof{I} \subseteq \{ (\ae^{\p}, \ae^{\q}) \in \esetof{
    {\apom^{\p}} } \times \esetof{ {\apom^{\q}} }, \p,\q \in \ptpset
  \sst \alfof{\apom^{\p}}(\ae^{\p}) = \aout,
  \alfof{\apom^{\q}}(\ae^{\q}) = \ain \}$.
\end{definition}
Informally, the inter-participant closure takes one pomset for every
participant and generates all \quo{acceptable} matches between output
and input events.
We use \cref{fig:ex:2-1} and \cref{fig:ex:2-2} to illustrate the
inter-participant closure.
The singleton $\aR_{\eqref{fig:ex:2-1}}$ contains one pomset that is
the composition of two independent pomsets:
$\apom_{\eqref{fig:ex:2-1}_a}$ and $\apom_{\eqref{fig:ex:2-1}_b}$.
Intuitively, this represents two concurrent \quo{threads} (hereafter
left and right threads) that have no interdependencies.
Let $\apom^{\p}$ be the projection of the single pomset in
$\aR_{\eqref{fig:ex:2-1}}$ for $\p \in \ptpset$, then the
inter-participant closure of $(\apom^{\p})_{\p \in \ptpset}$ consists
of the two well-formed pomsets of \cref{fig:ex:2-2}, the one that uses
the black and green dependencies, and the one that uses the black and
red dependencies.
Notice that the order $\leqof{I}$ in \cref{def:inter-closure} is a
subset of the product of outputs and matching inputs and this the
closure to contain  only well-formed pomsets.
For example, the closure of $\aR_{\eqref{fig:ex:2-1}}$ does not contain
the pomset having both green and red arrows.

\begin{figure}
  \center
  \begin{tikzpicture}[node distance = .5cm, scale=1, every node/.style = {rectangle,draw}, transform shape]\tiny
    \node (a1) at (0,0) {$\aout[@][C][][{\msg[l1]}]$};
    \node (a2) at (0,-2) {$\aout[@][@][][{\msg[x]}]$};
    \node (b1) at (1.5,-1) {$\aout[b][c][][{\msg[l2]}]$};
    \node (b2) at (1.5,-2) {$\ain[@][@][][{\msg[x]}]$};
    \node (b3) at (1.5,-3) {$\aout[b][c][][{\msg[l3]}]$};
    \node (c2) at (3,-1) {$\ain[b][c][][{\msg[l2]}]$};
    \node (c1) at (3,0) {$\ain[a][c][][{\msg[l1]}]$};
    \node (c3) at (3,-3) {$\ain[@][C][][{\msg[l3]}]$};
    \path[->,draw] (a1) -- (a2);
    \path[->,draw] (b1) -- (b2);
    \path[->,draw] (b2) -- (b3);
    \path[->,draw] (c1) -- (c2);
    \path[->,draw] (c2) -- (c3);
    \path[->,draw] (a1) -- (c1);
    \path[->,draw] (b1) -- (c2);
    \path[->,draw] (a2) -- (b2);
    \path[->,draw] (b3) -- (c3);
    \node (a1) at (4.5,0) {$\aout[@][C][][{\msg[r1]}]$};
    \node (a2) at (4.5,-2) {$\aout[@][@][][{\msg[x]}]$};
    \node (b1) at (6,-1) {$\aout[b][c][][{\msg[r2]}]$};
    \node (b2) at (6,-2) {$\ain[@][@][][{\msg[x]}]$};
    \node (b3) at (6,-3) {$\aout[b][c][][{\msg[r3]}]$};
    \node (c2) at (7.5,-1) {$\ain[b][c][][{\msg[r2]}]$};
    \node (c1) at (7.5,0) {$\ain[a][c][][{\msg[r1]}]$};
    \node (c3) at (7.5,-3) {$\ain[@][C][][{\msg[r3]}]$};
    \path[->,draw] (a1) -- (a2);
    \path[->,draw] (b1) -- (b2);
    \path[->,draw] (b2) -- (b3);
    \path[->,draw] (c1) -- (c2);
    \path[->,draw] (c2) -- (c3);
    \path[->,draw] (a1) -- (c1);
    \path[->,draw] (b1) -- (c2);
    \path[->,draw] (a2) -- (b2);
    \path[->,draw] (b3) -- (c3);
  \end{tikzpicture}
  \caption{$\aR_{\eqref{fig:ex:2-1}} = \{
    [
    \esetof{\apom_{\eqref{fig:ex:2-1}_a}} \cup \esetof{\apom_{\eqref{fig:ex:2-1}_b}},
    \leqof{\apom_{\eqref{fig:ex:2-1}_a}} \cup \leqof{\apom_{\eqref{fig:ex:2-1}_b}},
    \alfof{\apom_{\eqref{fig:ex:2-1}_a}} \cup \alfof{\apom_{\eqref{fig:ex:2-1}_b}}
    ]
    \}$}
  \label{fig:ex:2-1}
\end{figure}

\begin{figure}
  \center
  \begin{tikzpicture}[scale=1, every node/.style = {rectangle,draw}, transform shape]\tiny
    \node (a1) at (0,0) {$\aout[@][C][][{\msg[l1]}]$};
    \node (a2) at (0,-2) {$\aout[@][@][][{\msg[x]}]$};
    \node (a1') at (1.2,0) {$\aout[@][C][][{\msg[r1]}]$};
    \node (a2') at (1.2,-2) {$\aout[@][@][][{\msg[x]}]$};
    \node (b1) at (3,0) {$\aout[b][c][][{\msg[l2]}]$};
    \node (b2) at (3,-1) {$\ain[@][@][][{\msg[x]}]$};
    \node (b3) at (3,-2) {$\aout[b][c][][{\msg[l3]}]$};
    \node (b1') at (4.2,0) {$\aout[b][c][][{\msg[r2]}]$};
    \node (b2') at (4.2,-1) {$\ain[@][@][][{\msg[x]}]$};
    \node (b3') at (4.2,-2) {$\aout[b][c][][{\msg[r3]}]$};
    \node (c2) at (6,-1) {$\ain[b][c][][{\msg[l2]}]$};
    \node (c1) at (6,0) {$\ain[a][c][][{\msg[l1]}]$};
    \node (c3) at (6,-2.5) {$\ain[@][C][][{\msg[l3]}]$};
    \node (c2') at (7.2,-1) {$\ain[b][c][][{\msg[r2]}]$};
    \node (c1') at (7.2,0) {$\ain[a][c][][{\msg[r1]}]$};
    \node (c3') at (7.2,-2.5) {$\ain[@][C][][{\msg[r3]}]$};
    \path[->,draw] (a1) -- (a2);
    \path[->,draw] (a1') -- (a2');
    \path[->,draw] (b1) -- (b2);
    \path[->,draw] (b2) -- (b3);
    \path[->,draw] (b1') -- (b2');
    \path[->,draw] (b2') -- (b3');
    \path[->,draw] (c1) -- (c2);
    \path[->,draw] (c2) -- (c3);
    \path[->,draw] (c1') -- (c2');
    \path[->,draw] (c2') -- (c3');
    \path[->,draw] (a1) edge [bend left=12] (c1);
    \path[->,draw] (a1') edge [bend left=15] (c1');
    \path[->,draw] (b1) .. controls (4.5,-.7) .. (c2);
    \path[->,draw] (b1') .. controls (5,0) .. (c2');
    \path[->,draw] (b3) .. controls (4.5,-2.5) .. (c3);
    \path[->,draw] (b3') .. controls (5.7,-2) .. (c3');
    \path[->,draw,\newgreen] (a2) .. controls (1,-1.5) .. (b2);
    \path[->,draw,\newgreen] (a2'.east) .. controls (2,-1.5) .. (b2');
    \path[->,draw,red] (a2.south east) .. controls (3,-3) .. (b2'.south);
    \path[->,draw,red] (a2') .. controls (2,-1.5) .. (b2);
  \end{tikzpicture}
\caption{Inter-participant closure of pomset of \cref{fig:ex:2-1}}
\label{fig:ex:2-2}
\end{figure}
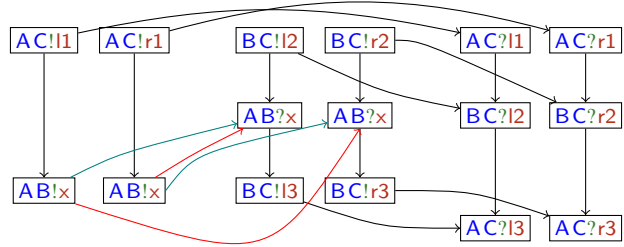

\begin{definition}
A \emph{pomset $\apom$ is \emph{less permissive} than pomset
  $\apom'$} (or \emph{$\apom'$ is more permissive than $\apom$},
written $\apom \rsubtype \apom'$) when
$\esetof{\apom} = \esetof{\apom'}$, $\alfof{\apom} = \alfof{\apom'}$,
and $\leqof{\apom} \supseteq \leqof{\apom'}$.
\end{definition}
\begin{lemma}
If $\apom \rsubtype \apom'$ then $\langof{\apom} \subseteq
\langof{\apom'}$. 
\end{lemma}

\begin{definition}[{\CCP[2]}]\label{def:ccp2}
  A set of pomsets $\aR$ over $\lset$ satisfies closure condition
  \CCP[2] if for all tuples $(\apom^{\p})_{\p \in \ptpset}$ of pomsets
  of $\aR$, for every pomset
  $\apom \in \normalform((\projpom{\apom^{\p}}{\p})_{\p \in
    \ptpset})$, there exists $\apom' \in \aR$ such that
  $\apom \sqsubseteq \apom'$.
\end{definition}
Intuitively, \cref{def:ccp2} requires that if all the possible
executions of a pomset cannot be distinguished by any of the
participants of $\aR$, then those executions must be part of the
language of $\aR$.
\cref{thm:ccp2vscc2} below shows that \CCP[2] entails \CC[2]; its
proof is based on \quo{counting} the number of events with a certain
label $\al$ preceding an event $\ae$ in the order $\leqof{\apom}$ of a
pomset $\apom$: we write $\countEvents$ for such number (namely,
$\countEvents$ is the cardinality of
$\{\ae' \in \esetof \apom \sst \ae' \leqof \apom \ae \land \alfof
\apom(\ae') = \al\}$).
\begin{theorem}\label{thm:ccp2vscc2}
  If $\aR$ satisfies \CCP[2] then $\langof \aR$ satisfies \CC[2].
\end{theorem}
\begin{proof}
  Let $\aW$ be a well-formed and complete word over $\lset$ that
  satisfies hypothesis of \CC[2]: for every participant
  $\p \in \ptpset$ there exists $\aW^{\p} \in \langof \aR$ for which
  $\projword{\aW}{\p} = \projword{\aW^{\p}}{\p}$.
  Then, for each $\p \in \ptpset$, there is a pomset
  $\apom^{\p} \in \aR$ such that a linearization $\lineariz$ of
  $\apom^{\p}$ yields $\aW^{\p}$.
  We can hence take the pomset
  \begin{align*}
    r= \left [
    \bigcup_{\p \in \ptpset} \esetof{\projpom{\apom^{\p}}{\p}},
    \quad
    \leqof{I} \cup
    \bigcup_{\p \in \ptpset}
    \leqof{\projpom{\apom^{\p}}{\p}},
    \quad
    \bigcup_{\p \in \ptpset} \alfof{\projpom{\apom^{\p}}{\p}}
    \right ] 
  \end{align*}
  where
  \begin{align*}
\leqof{I}  = 
    \bigcup_{\q \neq \p \in \ptpset}
    \left \{
    (\ae^{\p}, \ae^{\q})
    \in \esetof{
    \projpom{\apom^{\p}}{\p}
    } \times 
    \esetof{
    \projpom{\apom^{\q}}{\q}
    } \sst
    \begin{array}{l}
      \alfof{\apom^{\p}}(\ae^{\p}) = \aout 
      \textit{ and }\alfof{\apom^{\q}}(\ae^{\q}) = \ain \\
      \textit{and }\countEvents[\lineariz][\ae^{\p}][\aout] =
      \countEvents[{\lineariz[\q]}][\ae^{\q}][\ain]
    \end{array}
    \right \}
  \end{align*}
The pomset $\apom$ is in
$\normalform((\projpom{\apom^{\p}}{\p})_{\p \in \ptpset})$,
since it is well-formed and complete and
$\leqof{I}$ satisfies conditions of Definition~\ref{def:inter-closure}.
In fact, since $\aW$
is well-formed and complete, all send and receive events have
corresponding matching events.
Also by construction, $\aW \in \langof{\apom}$ and, for every $\p$,
$\projpom{\apom}{\p} \rsubtype \projpom{\apom^{\p}}{\p}$.
Finally, by \CCP[2] there exists $\apom' \in \aR$ such that
$\apom \sqsubseteq \apom'$, therefore $\aW \in \langof{\apom'}$ hence
$\aW \in \langof{\aR}$.
\end{proof}

\cref{fig:ex:2-1} provides an example of a family of pomsets that cannot
be weakly realised.
An execution of this specification can be as follows:
\begin{enumerate}
\item the left thread of $\p$ executes $\aout[\p][C][][{\msg[l_1]}]$
  and $\aout[\p][\q][][{\msg[x]}]$
\item the right thread of $\q$ executes $\aout[\q][C][][{\msg[r_2]}]$
  and $\ain[\p][\q][][{\msg[x]}]$, \quo{stealing} the message
  $\msg[x]$ generated by the left thread of $\p$ and meant for the
  left thread of $\q$
\item the right thread of $\q$ executes
  $\aout[\q][C][][{\msg[r_3]}]$.
\end{enumerate}
This violates the constraint that event $\aout[\p][c][][{\msg[r_1]}]$
must always precede event $\aout[\q][C][][{\msg[r_3]}]$, which the
specification imposes independently of the interleaved execution of
the participants' threads.
Indeed, $\aR_{\eqref{fig:ex:2-1}}$ does not satisfy \CCP[2].
In fact, there are two well-formed and complete pomsets that satisfy
the hypothesis of \CCP[2]: the pomset of \cref{fig:ex:2-2} that uses
the black and green dependencies, and the one that uses the black and
red dependencies.
Condition \CCP[2] is violated because there is no pomset in
$\aR_{\eqref{fig:ex:2-1}}$ that is more permissive than the pomset
using the red dependencies.

The next condition requires to introduce the concept of \emph{prefix}
of a pomset $\apom$, which is a pomset $\apom'$ on a subset of the
events of $\apom$ that preserves the order and labelling of $\apom$;
formally (following~\cite{katoen1998pomsets})
\begin{definition}[Prefix pomsets]\label{def:prefixpomset}
  A pomset $\apom'=[\eset', \leq', \alf']$ is a \emph{prefix} of
  pomset $\apom=[\eset, \leq, \alf]$ if there exists a label
  preserving injection $\phi:\eset' \rightarrow \eset$ such that
  $\phi(\leq') = \leq \cap (\eset \times \phi(\eset'))$
\end{definition}
We remark that an arbitrary sub-pomset satisfies the weaker condition
$\phi(\leq') = \leq \cap (\phi(\eset') \times \phi(\eset'))$.
Instead, $\phi(\leq') = \leq \cap (\eset \times \phi(\eset'))$
prevents events in
$\eset \setminus \phi(\eset')$ from preceding events in $\phi(\eset')$ and it is
equivalent to say that for all $\ae' \in \eset'$ if there is
$\ae \leq \phi(\ae')$ then there exists $\ae'' \in \eset'$ such that $\phi(\ae'')=\ae$ and $\ae'' \leq' \ae'$.

\begin{lemma}
  Let $\apom$ be a pomset over $\lset$ and $\aW$ be a word in
  $\lset^{\star}$,
  $\aW \in \pref[{\langof{\apom}}]$
  if, and only if,
  there exists a prefix $\apom'$ of $\apom$ such that
  $\aW \in \langof{\apom'}$.
\end{lemma}

\begin{definition}[{\CCP[3]}]
  A set of pomsets $\aR$ over $\lset$ satisfies closure condition
  \CCP[3] if for all tuples of pomsets
  $(\bar \apom^{\p})_{\p \in \ptpset}$ such that for every $\p$
  $\bar \apom^{\p}$ is a prefix of a pomset $\apom^{\p} \in \aR$, and
  for every pomset
  $\bar \apom \in \normalform((\projpom{\bar \apom^{\p}}{\p})_{\p \in
    \ptpset})$ there is a pomset $\apom' \in \aR$ and a prefix
  $\bar \apom'$ of $\apom'$ such that
  $\bar \apom \rsubtype \bar \apom'$.
\end{definition}
\begin{theorem}\label{thm:cc3pom}
  If $\aR$ satisfies \CCP[3] then $\langof \aR$ satisfies \CC[3].
\end{theorem}
\begin{proof}
  Let $\aW$ be a word that satisfies hypothesis of \CC[3]: for every
  participant $\p \in \pset$, there exists a word
  $\aW^{\p} \in \pref[\langof{\aR}]$ such that
  $\projword{\aW}{\p} = \projword{\aW^{\p}}{\p}$.  Therefore,
  there is a pomset $\bar \apom^{\p}$ 
  prefix of a pomset $\apom^{\p} \in \aR$ such that
  $\aW^{\p} \in \langof{\bar \apom^{\p}}$  and let $\lineariz$ be one of the
  linearizations of $\bar \apom^{\p}$ that corresponds to $\aW^{\p}$.
  Define
  \begin{align*}
    \bar \apom = \left [
    \bigcup_{\p \in \ptpset} \esetof{\projpom{\bar \apom^{\p}}{\p}},
    \qquad
    \leqof{I} \cup
    \bigcup_{\p \in \ptpset}
    \big(
    \leqof{\projpom{\bar  \apom^{\p}}{\p}}
    \big),
    \quad
    \bigcup_{\p \in \ptpset} \alfof{\projpom{\bar \apom^{\p}}{\p}},
    \right ]
  \end{align*}
where
  \begin{align*}
\leqof{I}  = 
    \bigcup_{\q \neq \p \in \ptpset}
    \left \{
    (\ae^{\p}, \ae^{\q})
    \in \esetof{
    \projpom{\bar \apom^{\p}}{\p}
    } \times 
    \esetof{
    \projpom{\bar \apom^{\q}}{\q}
    } \sst
    \begin{array}{l}
      \alfof{\bar \apom^{\p}}(\ae^{\p}) = \aout 
      \textit{ and }\alfof{\bar \apom^{\q}}(\ae^{\q}) = \ain \\
      \textit{and }\countEvents[\lineariz][\ae^{\p}][\aout] =
      \countEvents[{\lineariz[\q]}][\ae^{\q}][\ain]
    \end{array}
    \right \}
  \end{align*}
The pomset $\bar \apom$ is in
$\normalform((\projpom{\bar \apom^{\p}}{\p})_{\p \in \ptpset})$,
since it is well-formed and
$\leqof{I}$ satisfies conditions of Definition~\ref{def:inter-closure}.
In fact, since $\aW$ is
well-formed, all receives have matching sends.
Also by construction,  $\aW \in \langof{\bar \apom}$
and, for every $\p$,
$\projpom{\bar \apom}{\p} \rsubtype \projpom{\bar \apom^{\p}}{\p}$.
Hence, by \CCP[3] there exists $\apom' \in \aR$ and a prefix
$\bar \apom'$ of $\apom $ such that $\bar \apom \sqsubseteq \bar \apom'$,
therefore $\aW \in \langof{\bar \apom'}$ and therefore
$\aW \in \pref[\langof{\aR}]$.
\end{proof}
From Theorems~\ref{thm:cc3},\ref{thm:ccp2vscc2}, and \ref{thm:cc3pom},
it follows that if a set of pomsets $\aR$ satisfies \CCP[2] and
\CCP[3] then $\langof \aR$ is safe realisable.

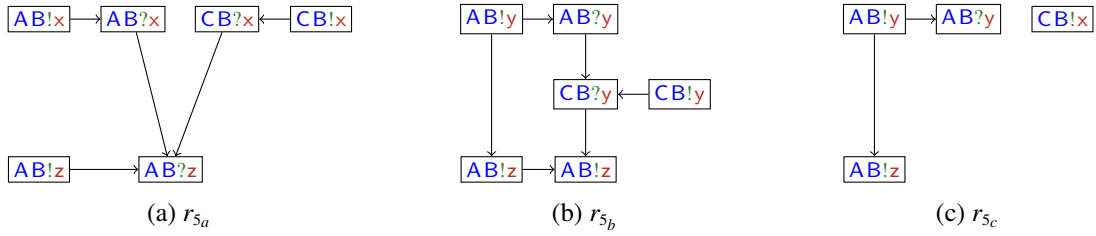
\begin{figure*}[t!]
    \centering

%  LocalWords:  realisable
    \begin{subfigure}[b]{0.35\textwidth}
      \centering
      \begin{tikzpicture}[scale=1, every node/.style = {rectangle,draw}, transform shape]\tiny
        \node (a1) at (0,0) {$\aout[@][@][][{\msg[x]}]$};
        \node (a2) at (0,-2) {$\aout[@][@][][{\msg[z]}]$};
        \node (b1) at (1.25,0) {$\ain[@][@][][{\msg[x]}]$};
        \node (b2) at (2.5,0) {$\ain[c][@][][{\msg[x]}]$};
        \node (b3) at (1.75,-2) {$\ain[@][@][][{\msg[z]}]$};
        \node (c1) at (3.75,0) {$\aout[c][b][][{\msg[x]}]$};
        %\path[->,draw] (a1) -- (a2);
        \path[->,draw] (b1) -- (b3);
        \path[->,draw] (b2) -- (b3);
        \path[->,draw] (a1) -- (b1);
        \path[->,draw] (c1) -- (b2);
        \path[->,draw] (a2) -- (b3);
        \end{tikzpicture}      
        \caption{$\apom_{\ref{fig:example:cc3pom}_a}$}
    \end{subfigure}%
    ~ 
    \begin{subfigure}[b]{0.3\textwidth}
        \centering
        \begin{tikzpicture}[scale=1, every node/.style = {rectangle,draw}, transform shape]\tiny
          \node (a1) at (0,0) {$\aout[@][@][][{\msg[y]}]$};
          \node (a2) at (0,-2) {$\aout[@][@][][{\msg[z]}]$};
          \node (b1) at (1.25,0) {$\ain[@][@][][{\msg[y]}]$};
          \node (b2) at (1.25,-1) {$\ain[c][@][][{\msg[y]}]$};
          \node (b3) at (1.25,-2) {$\aout[@][@][][{\msg[z]}]$};
          \node (c1) at (2.5,-1) {$\aout[c][b][][{\msg[y]}]$};
          \path[->,draw] (a1) -- (a2);
          \path[->,draw] (b1) -- (b2);
          \path[->,draw] (b2) -- (b3);
          \path[->,draw] (c1) -- (b2);
          \path[->,draw] (a1) -- (b1);
          \path[->,draw] (a2) -- (b3);
        \end{tikzpicture}
        \caption{$\apom_{\ref{fig:example:cc3pom}_b}$}
    \end{subfigure}
    ~ 
    \begin{subfigure}[b]{0.3\textwidth}
      \centering
      \begin{tikzpicture}[scale=1, every node/.style = {rectangle,draw}, transform shape]\tiny
        \node (a1) at (0,0) {$\aout[@][@][][{\msg[y]}]$};
        \node (a2) at (0,-2) {$\aout[@][@][][{\msg[z]}]$};
        \node (b1) at (1.25,0) {$\ain[@][@][][{\msg[y]}]$};
        \node (c1) at (2.5,0) {$\aout[c][b][][{\msg[x]}]$};
        \path[->,draw] (a1) -- (a2);
        \path[->,draw] (a1) -- (b1);
      \end{tikzpicture}
      \caption{$\apom_{\ref{fig:example:cc3pom}_c}$}
    \end{subfigure}
    \caption{\label{fig:example:cc3pom}The language of
      $\{\apom_{\ref{fig:example:cc3pom}_a},
      \apom_{\ref{fig:example:cc3pom}_b}\}$ is not realisable}
\end{figure*}

The family of pomsets
$\aR = \{\apom_{\ref{fig:example:cc3pom}_a},
\apom_{\ref{fig:example:cc3pom}_b}\}$
of \cref{fig:example:cc3pom} exemplifies a common obstacle for safe
realisability.  Here, participants $\p$ and
$\ptp[C]$ should both send the message $\msg[x]$ or both send the
message $\msg[y]$.
However, $\p$ and $\ptp[C]$ do not coordinate to achieve this
behaviour; this makes it impossible for them to distributively commit
to a common choice.
The family of pomsets $\aR$ does not satisfy 
\CCP[3].
In fact, pomset $\apom_{\ref{fig:example:cc3pom}_c}$ satisfies hypothesis of \CCP[3]
(using $\apom_{\ref{fig:example:cc3pom}_a}$ for $\ptp[C]$ and $\apom_{\ref{fig:example:cc3pom}_b}$ for both $\p$
and $\q$), however there is no pomset in $\aR$ whose prefix
is more permissive that $\apom_{\ref{fig:example:cc3pom}_c}$.

\newcommand{\pomsetTerm}[1][\ptpset']{\mathcal{R}_{#1}}
Like for the closure conditions, we lift the sufficient condition
for termination soundness to pomsets.
\begin{definition}[Terminating pomsets]
  A participant $\p \in \pset$ is \emph{termination-unaware} for a set
  of pomsets $\aR$ if there are $\apom, \apom' \in \aR$, and a
  label-preserving injection
  $\phi:\esetof{\projpom{\apom}{\p}} \to
  \esetof{\projpom{\apom'}{\p}}$
  such that
  $\leq = \phi(\leqof{\projpom{\apom}{\p}}) \cup
  \leqof{\projpom{\apom'}{\p}}$ is a partial order and
  \[
  \mathit{min}_{\leq}(\esetof{\projpom{\apom'}{\p}}) \subseteq
  \phi(\mathit{min}_{\leqof{\projpom{\apom}{\p}}}(\esetof{\projpom{\apom}{\p}}))
  \qqand
  \mathit{min}_\leq( \esetof{\projpom{\apom'}{\p}} \setminus
  \phi(\esetof{\projpom{\apom}{\p}}) ) \cap \lset^{?} \neq \emptyset
  \]
  Given a set of participants $\ptpset' \subseteq \ptpset$, we say
  that $\aR$ is \emph{$\ptpset'$-terminating} when there is no $\p \in \ptpset'$
  termination-unaware for $\aR$.
\end{definition}

\begin{figure*}[t!]
    \centering
    \begin{subfigure}[b]{0.35\textwidth}
      \centering
      \begin{tikzpicture}[scale=1, every node/.style = {rectangle,draw}, transform shape]\tiny
        \node (a1) at (0,0) {$\aout[@][@][][{\msg[x]}]$};
        \node (a2) at (0,-1) {$\aout[@][@][][{\msg[y]}]$};
        \node (a3) at (0,-2) {$\aout[@][@][][{\msg[z]}]$};
        \node (b1) at (2,0) {$\ain[@][@][][{\msg[x]}]$};
        \node (b2) at (2,-1) {$\ain[@][@][][{\msg[y]}]$};
        \node (b3) at (2,-2) {$\ain[@][@][][{\msg[z]}]$};
        \path[->,draw] (a1) -- (a2);
        \path[->,draw] (a1) -- (b1);
        \path[->,draw] (a2) -- (a3);
        \path[->,draw] (a2) -- (b2);
        \path[->,draw] (b1) -- (b2);
        \path[->,draw] (b2) -- (b3);
        \path[->,draw] (a3) -- (b3);
      \end{tikzpicture}
      \caption{$\apom_{\ref{fig:example:term2}_a}$}
    \end{subfigure}%
    ~ 
    \begin{subfigure}[b]{0.35\textwidth}
      \centering
      \begin{tikzpicture}[scale=1, every node/.style = {rectangle,draw}, transform shape]\tiny
        \node (a1) at (0,0) {$\aout[@][@][][{\msg[x]}]$};
        \node (a2) at (0,-1) {$\aout[@][@][][{\msg[y]}]$};
        \node (a3) at (0,-2) {$\aout[@][@][][{\msg[z]}]$};
        \node (a4) at (0,-3) {$\aout[@][@][][{\msg[w]}]$};
        \node (b1) at (3,0) {$\ain[@][@][][{\msg[x]}]$};
        \node (b2) at (2,-1) {$\ain[@][@][][{\msg[y]}]$};
        \node (b3) at (3,-2) {$\ain[@][@][][{\msg[z]}]$};
        \node (b4) at (2,-3) {$\ain[@][@][][{\msg[w]}]$};
        \path[->,draw] (a1) -- (a2);
        \path[->,draw] (a1) -- (b1);
        \path[->,draw] (a2) -- (a3);
        \path[->,draw] (a2) -- (b2);
        \path[->,draw] (a3) -- (a4);
        \path[->,draw] (b1) -- (b2);
        \path[->,draw] (b1) -- (b3);
        \path[->,draw] (b2) -- (b4);
        \path[->,draw] (a3) -- (b3);
        \path[->,draw] (a4) -- (b4);
        \path[->,draw] (b3) -- (b4);
      \end{tikzpicture}
        \caption{$\apom_{\ref{fig:example:term2}_b}$}
    \end{subfigure}
    ~ 
    \begin{subfigure}[b]{0.25\textwidth}
      \centering
      \begin{tikzpicture}[scale=1, every node/.style = {rectangle,draw}, transform shape]\tiny
        \newcounter{row}{0}
        \foreach \m in {x,y,z,w}{
          \refstepcounter{row}
          \node (out\m) at (0,-\therow) {$\aout[@][@][][{\msg[\m]}]$};
          \node (in\m) at (2,-\therow) {$\ain[@][@][][{\msg[\m]}]$};
          \path[->,draw] (out\m) -- (in\m);
        }
        \foreach \s/\t in {x/y,y/z,z/w}{
          \path[->,draw] (out\s) -- (out\t);
          \path[->,draw] (in\s) -- (in\t);
        }
        \node[fill=orange!50,opacity=.3,] (draw) at (2,-4) {\phantom{$\ain[@][@][][{\msg[w]}]$}};
      \end{tikzpicture}
        \caption{$\leqof{\apom_{\ref{fig:example:term2}_a}} \cup \leqof{\apom_{\ref{fig:example:term2}_a}}$}
    \end{subfigure}
    \caption{\label{fig:example:term2}The set $\aR_{\eqref{fig:example:term2}}=\{\apom_{\ref{fig:example:term2}_a}, \apom_{\ref{fig:example:term2}_a}\}$ is not termination
      sound for $\q$}
\end{figure*}
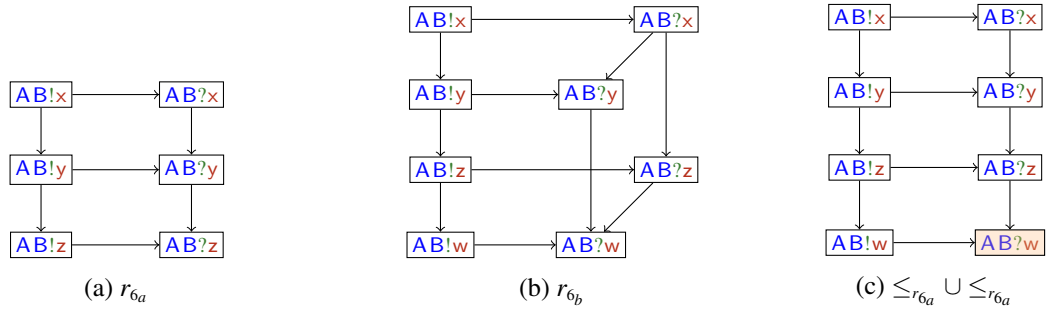
We use \cref{fig:example:term2} to describe termination awareness.
$\q$ is \emph{termination-unaware} for
 the set of pomsets
$\aR_{\eqref{fig:example:term2}}$.
In fact,
let
$\phi:\esetof{\projpom{\apom_{\ref{fig:example:term2}_a}}{\q}} \to
\esetof{\projpom{\apom_{\ref{fig:example:term2}_b}}{\q}}$ be the only
possible label-preserving injection, then
$\leq= \phi(\leqof{\projpom{\apom_{\ref{fig:example:term2}_a}}{\q}}) \cup
\leqof{\projpom{\apom_{\ref{fig:example:term2}_b}}{\q}}$ is the
partial order in \cref{fig:example:term2}.c, 
and
$\mathit{min}_\leq( \esetof{\projpom{\apom_{\ref{fig:example:term2}_b}}{\q}} \setminus
  \phi(\esetof{\projpom{\apom_{\ref{fig:example:term2}_a}}{\q}}) )  =
\{\ain[\p][\q][][{\msg[w]}]\}$  is not disjoint from $\lset^{?}$.
Intuitively, 
$\leq$
represents the intersection of the languages of the two pomsets
$\projpom{\apom_{\ref{fig:example:term2}_b}}{\q}$ and 
$\projpom{\apom_{\ref{fig:example:term2}_a}}{\q}$.
\begin{theorem}
  Given $\ptpset' \subseteq \ptpset$, if $\aR$ is
  $\ptpset'$\emph{-terminating} then $\langof \aR$ is $\ptpset'$\emph{-terminating}.
\end{theorem}
\begin{proof}
Given a word $\aW \in \langof{\aR}$, there is a pomset
$\apom \in \aR$ such that $\aW \in \langof{\aR}$.
Let $\p \in \ptpset'$ 
and assume that there is 
$\aW' \in \langof{\aR}$ such that 
$\projword{\aW}{\p}$ is a prefix of $\projword{\aW'}{\p}$.
Therefore, there is a pomset $\apom' \in \aR$
such that $\projword{\aW'}{\p} \in \langof{\projpom{\apom'}{\p}}$.
Let $\ae_1, \dots, \ae_n$ and 
$\ae'_1, \dots, \ae'_{n'}$, with $n < n'$, be the linearizations
of $\leq_{\apom}$ and $\leq_{\apom'}$ respectively for the world $\aW$ and $\aW'$
respectively.
Let $\phi$ be the injection that maps $\ae_i$ to $\ae'_i$ for $1 \leq i \leq n$,
then $\leq = 
\phi(\leqof{\projpom{\apom}{\p}}) \cup
\leqof{\projpom{\apom'}{\p}}$ is a partial order.
Therefore
$\mathit{min}_\leq( \esetof{\projpom{\apom'}{\p}} \setminus
  \phi(\esetof{\projpom{\apom}{\p}}) ) \cap \lset^{?} \neq \emptyset$ 
 since $\aR$ is $\ptpset'$\emph{-terminating},
thus the first symbol of $\aW'$ after $\aW$ cannot be
an input.
\end{proof}

%%% Local Variables:
%%% mode: latex
%%% TeX-master: "main"
%%% End:

%% file: implementation.tex
If a pomset is thought of as the specification of a possible scenario
of a system, a practical advantage of using the conditions of
\cref{sec:pomsets} is that problems can be discovered at design-time.
This permits to easily isolate the problematic scenarios of a
specification even if they share multiple traces with non-problematic
scenarios.

Checking \CCP[2] and \CCP[3] is decidable since we assume $\aR$ to be a finite
set of finite pomsets and $\ptpset$ to be finite.
For \CCP[2], there are finite tuples $(\apom^{\p})_{\p \in \ptpset}$ of pomsets of
$\aR$ and for each tuple the inter-participant closure is a finite set of
finite pomsets. For \CCP[3], the number of prefixes of pomsets in $\aR$ is also 
finite.
However, verifying these conditions is in general 
expensive due to two reasons: the combinatorial explosion of the
inter-participant closure and the need of finding a graph isomorphism
to check relation $\rsubtype$ between pomsets and to prove the
existence of the label preserving injection $\phi$.
In both cases,
this complexity depends on the presence of multiple and independent
instances of the same action.
\begin{definition}
  Let $\apom$ be a pomset over $\lset$. An action $\al \in \lset$ 
  \emph{concurrently repeats in $\apom$} if there exist $\ae,\ae' \in \esetof{\apom}$
  such that $\ae \neq \ae'$,
  $\alfof{\apom}(\ae)=\alfof{\apom}(\ae')=\al$,
  and neither $\ae \leqof{\apom} \ae'$ nor $\ae' \leqof{\apom} \ae$.
\end{definition}
In practice, the presence of actions that concurrently repeat is
limited. In fact, specification formalisms usually impose conditions
that syntactically avoid this issue (e.g. see well-forkedness of
~\cite{gt18} or the even more restrictive conditions of
e.g.,~\cite{honda16jacm}) because sending the same message in two
independent threads may \quo{confuse} receivers making it hard (or
impossible) to decide which receiving thread should consume the
message, leading to coordination problems.

We sketch the complexity analysis for \CCP[2].
For a set of pomsets $\aR$, there are $|\aR|^{|\ptpset|}$ possible
tuples $(\apom^{\p})_{\p \in \ptpset}$.  For each tuple
$(\apom^{\p})_{\p \in \ptpset}$, the number of pomsets in the
inter-participant closure is proportional to $\prod_e 2^{\#(e)}$,
where $\#(e)$ is the number of concurrent repetitions of the action of
an event $e$ in $(\projpom{\apom^{\p}}{\p})_{\p \in \ptpset})$.
Therefore, if there are no concurrently repeated actions then the
inter-participant closure contains at most one pomset.
Checking $\apom \sqsubseteq \apom'$ requires to find a label
preserving injection $\phi$ from events of $\apom$ to events of
$\apom'$ that does not violate event orders.
This problem can be reduced to graph isomorphism and its complexity is
exponential in $\prod_e {\#(e)}$.
In fact, for every pomset $\apom$ in the inter-participant closure,
the restriction to the events of $\apom$ having a same
non-concurrently repeated action is totally ordered by the order of
$\apom$, thus the identification of the injection is trivial.
Therefore, if there are no concurrently repeated actions in $\aR$ then
checking \CCP[2] can be done in polynomial time with respect to the
number of events.  Condition \CCP[2] avoids the explicit computation
of the language of the family of pomsets, which can lead to
combinatorial explosion due to interleavings.

For example, $\aR_{\eqref{fig:ex:2-1}}$ contains one pomset and has two actions that occur
concurrently: $\aout[@][@][][{\msg[x]}]$ and $\aout[\q][\p][][{\msg[x]}]$.
Therefore there is only one tuple $(\apom^{\p})_{\p \in \ptpset}$ and
its inter-participant closure has two pomsets (see Fig.\ref{fig:ex:2-2}).
Checking $\sqsubseteq$ between these pomsets and the pomset in
$\aR_{\eqref{fig:ex:2-1}}$, requires to iterate over all possible label preserving isomorphisms.
However, since all actions except  $\aout[@][@][][{\msg[x]}]$ and
$\aout[\q][\p][][{\msg[x]}]$ do not occur concurrently, there are only two
of such isomorphisms.
Checking \CC[2] can be more expensive. Pomsets $\apom_{\eqref{fig:ex:2-1}_a}$ and
$\apom_{\eqref{fig:ex:2-1}_b}$ of \cref{fig:ex:2-1} have $32$
different linearizations, each one consisting of $8$ events.
Therefore the language of $\aR_{\eqref{fig:ex:2-1}}$ consists of
$32 * 32 * 2^{8} = 2^{18}$ words.  Therefore, directly analyzing the
inter-participant closure in \cref{fig:ex:2-2} is more efficient.

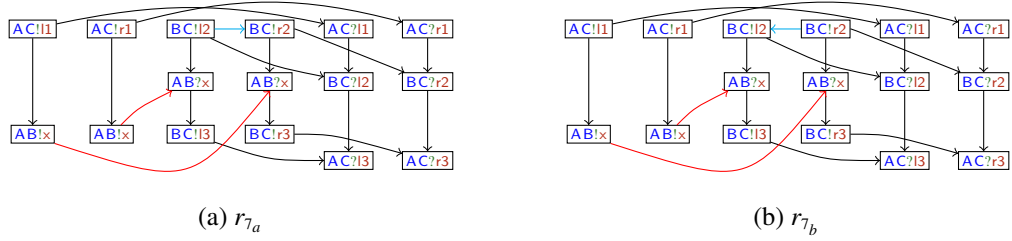
\begin{figure*}[t!]
    \centering
    \begin{subfigure}[b]{0.45\textwidth}
      \centering
      \begin{tikzpicture}[scale=.7, every node/.style = {rectangle,draw}, transform shape]\tiny
        \node (a1) at (0,0) {$\aout[@][C][][{\msg[l1]}]$};
        \node (a2) at (0,-2) {$\aout[@][@][][{\msg[x]}]$};
        \node (a1') at (1.5,0) {$\aout[@][C][][{\msg[r1]}]$};
        \node (a2') at (1.5,-2) {$\aout[@][@][][{\msg[x]}]$};
        \node (b1) at (3,0) {$\aout[b][c][][{\msg[l2]}]$};
        \node (b2) at (3,-1) {$\ain[@][@][][{\msg[x]}]$};
        \node (b3) at (3,-2) {$\aout[b][c][][{\msg[l3]}]$};
        \node (b1') at (4.5,0) {$\aout[b][c][][{\msg[r2]}]$};
        \node (b2') at (4.5,-1) {$\ain[@][@][][{\msg[x]}]$};
        \node (b3') at (4.5,-2) {$\aout[b][c][][{\msg[r3]}]$};
        \node (c2) at (6,-1) {$\ain[b][c][][{\msg[l2]}]$};
        \node (c1) at (6,0) {$\ain[a][c][][{\msg[l1]}]$};
        \node (c3) at (6,-2.5) {$\ain[@][C][][{\msg[l3]}]$};
        \node (c2') at (7.5,-1) {$\ain[b][c][][{\msg[r2]}]$};
        \node (c1') at (7.5,0) {$\ain[a][c][][{\msg[r1]}]$};
        \node (c3') at (7.5,-2.5) {$\ain[@][C][][{\msg[r3]}]$};
        \path[->,draw] (a1) -- (a2);
        \path[->,draw] (a1') -- (a2');
        \path[->,draw] (b1) -- (b2);
        \path[->,draw,cyan] (b1) -- (b1');
        \path[->,draw] (b2) -- (b3);
        \path[->,draw] (b1') -- (b2');
        \path[->,draw] (b2') -- (b3');
        \path[->,draw] (c1) -- (c2);
        \path[->,draw] (c2) -- (c3);
        \path[->,draw] (c1') -- (c2');
        \path[->,draw] (c2') -- (c3');
        \path[->,draw] (a1) edge [bend left=12] (c1);
        \path[->,draw] (a1') edge [bend left=15] (c1');
        \path[->,draw] (b1) .. controls (4.5,-.7) .. (c2);
        \path[->,draw] (b1') .. controls (5,0) .. (c2');
        \path[->,draw] (b3) .. controls (4.5,-2.5) .. (c3);
        \path[->,draw] (b3') .. controls (5.7,-2) .. (c3');
        \path[->,draw,red] (a2.south east) .. controls (3,-3) .. (b2'.south);
        \path[->,draw,red] (a2') .. controls (2,-1.5) .. (b2);
      \end{tikzpicture}
      \caption{$\apom_{\ref{fig:ex:counter}_a}$}
    \end{subfigure}%
    ~
    \begin{subfigure}[b]{0.45\textwidth}
        \centering
      \begin{tikzpicture}[scale=.7, every node/.style = {rectangle,draw}, transform shape]\tiny
        \node (a1) at (0,0) {$\aout[@][C][][{\msg[l1]}]$};
        \node (a2) at (0,-2) {$\aout[@][@][][{\msg[x]}]$};
        \node (a1') at (1.5,0) {$\aout[@][C][][{\msg[r1]}]$};
        \node (a2') at (1.5,-2) {$\aout[@][@][][{\msg[x]}]$};
        \node (b1) at (3,0) {$\aout[b][c][][{\msg[l2]}]$};
        \node (b2) at (3,-1) {$\ain[@][@][][{\msg[x]}]$};
        \node (b3) at (3,-2) {$\aout[b][c][][{\msg[l3]}]$};
        \node (b1') at (4.5,0) {$\aout[b][c][][{\msg[r2]}]$};
        \node (b2') at (4.5,-1) {$\ain[@][@][][{\msg[x]}]$};
        \node (b3') at (4.5,-2) {$\aout[b][c][][{\msg[r3]}]$};
        \node (c2) at (6,-1) {$\ain[b][c][][{\msg[l2]}]$};
        \node (c1) at (6,0) {$\ain[a][c][][{\msg[l1]}]$};
        \node (c3) at (6,-2.5) {$\ain[@][C][][{\msg[l3]}]$};
        \node (c2') at (7.5,-1) {$\ain[b][c][][{\msg[r2]}]$};
        \node (c1') at (7.5,0) {$\ain[a][c][][{\msg[r1]}]$};
        \node (c3') at (7.5,-2.5) {$\ain[@][C][][{\msg[r3]}]$};
        \path[->,draw] (a1) -- (a2);
        \path[->,draw] (a1') -- (a2');
        \path[->,draw] (b1) -- (b2);
        \path[->,draw,cyan] (b1') -- (b1);
        \path[->,draw] (b2) -- (b3);
        \path[->,draw] (b1') -- (b2');
        \path[->,draw] (b2') -- (b3');
        \path[->,draw] (c1) -- (c2);
        \path[->,draw] (c2) -- (c3);
        \path[->,draw] (c1') -- (c2');
        \path[->,draw] (c2') -- (c3');
        \path[->,draw] (a1) edge [bend left=12] (c1);
        \path[->,draw] (a1') edge [bend left=15] (c1');
        \path[->,draw] (b1) .. controls (4.5,-.7) .. (c2);
        \path[->,draw] (b1') .. controls (5,0) .. (c2');
        \path[->,draw] (b3) .. controls (4.5,-2.5) .. (c3);
        \path[->,draw] (b3') .. controls (5.7,-2) .. (c3');
        \path[->,draw,red] (a2.south east) .. controls (3,-3) .. (b2'.south);
        \path[->,draw,red] (a2') .. controls (2,-1.5) .. (b2);
      \end{tikzpicture}
        \caption{$\apom_{\ref{fig:ex:counter}_b}$}
    \end{subfigure}
    \caption{\label{fig:ex:counter}
A set of pomsets language-equivalent to 
the pomset with red and black dependencies of \cref{fig:ex:2-2},
but explicitly interleaves the events $\aout[\q][C][][{\msg[l2]}]$
and $\aout[\q][C][][{\msg[r2]}]$ (cyan dependencies)}
\end{figure*}

We remark that the conditions of \cref{sec:pomsets} strictly entail
the corresponding ones in \cref{sec:realisability}.
We show a counterexample for \CCP[2] only, since the same reasoning
applies for the other condition.
Consider the set
$\aR_{\ref{fig:ex:2-2}} = \{ \apom_{\ref{fig:ex:2-2}_{red}},
\apom_{\ref{fig:ex:2-2}_{green}} \}$, where 
$\apom_{\ref{fig:ex:2-2}_{red}}$ and
$\apom_{\ref{fig:ex:2-2}_{green}}$ respectively are the pomset with red dependencies
and the pomset with green dependencies of Figure~\ref{fig:ex:2-2}.
Then, $\aR_{\ref{fig:ex:2-2}}$
 satisfies \CCP[2], since it
contains all pomsets that satisfy hypothesis of the closure condition,
therefore by Theorem~\ref{thm:ccp2vscc2} its language satisfies
\CC[2].  Consider the set
$\aR_{\ref{fig:ex:counter}} = \{\apom_{\ref{fig:ex:counter}_a},
\apom_{\ref{fig:ex:counter}_b}, \apom_{\ref{fig:ex:2-2}_{green}} \}$,
where $\apom_{\ref{fig:ex:counter}_a}$ and $\apom_{\ref{fig:ex:counter}_b}$
are the two pomsets of
Figure~\ref{fig:ex:counter}.  Notice that $\apom_{\ref{fig:ex:counter}_a}$
and $\apom_{\ref{fig:ex:counter}_b}$ are equivalent to
$\apom_{\ref{fig:ex:2-2}_{red}}$, with the exception of the dependency
between $\aout[\q][C][][{\msg[l2]}]$ and $\aout[\q][C][][{\msg[r2]}]$.
Since $\apom_{\ref{fig:ex:counter}_a}$ and
$\apom_{\ref{fig:ex:counter}_b}$ have opposite orders between these two
events, the union of their languages is equal to the language of
$\apom_{\ref{fig:ex:2-2}_{red}}$.  Therefore the language of
$\aR_{\ref{fig:ex:counter}}$ is equal to the language of
$\aR_{\ref{fig:ex:2-2}}$, hence it also satisfies \CC[2].  However,
$\aR_{\ref{fig:ex:counter}}$ does not satisfy \CCP[2]. In fact, the
pomset $\apom_{\ref{fig:ex:2-2}_{red}}$ satisfies hypothesis of
\CCP[2], but there is not pomset in $\aR_{\ref{fig:ex:counter}}$ that
is more permissive than $\apom_{\ref{fig:ex:2-2}_{red}}$.

%%% Local Variables:
%%% mode: latex
%%% TeX-master: "main"
%%% End:

%% file: related.tex
The surge of \emph{message-passing} applications in industry is
revamping the interest for software engineering methodologies
supporting designers and developers called to realise
communication-centred software.
In this context, realisability of global specifications is of concern
for both practical and theoretical reasons.
Our approach can support choreography languages (e.g. the global
graphs used in~\cite{gt18} that allow multi-threaded participants and
complex distributed choices).
These specifications yield at the same time ($i$) concrete support to
scenario-based development, ($ii$) rigorous semantics in terms of
partial order of communication events that enable the use of
algorithms and tools to reason about and verify communicating
applications, and ($iii$) a simple graphical syntax that supports the
intuition and makes it easy to practitioners to master the
specification without needing to delve into the underlying theory.

A paradigmatic class of such formalisms are \emph{message-sequence
  charts} (MSCs)~\cite{itu11,gb13,mp05,hm03,gmptw01,ahptb96}.
A mechanism to statically detect realisability in MSCs is proposed
in~\cite{ben1997syntactic}.
The notions of non-local choices and of termination considered
in~\cite{ben1997syntactic} are less than than our verification conditions 
since intra-participant concurrency is not allowed and termination awareness
(\cref{def:termination}) is not enforced.
In the context of choreographies, several works
(e.g.,~\cite{bmt14,chy07,honda16jacm}) defined constraints to
guarantee the soundness of the projections of global specifications.
These approaches address the problem for specific languages, thus
these conditions often use information on the syntactical structure of
the specification.
Instead, conditions presented in \cref{sec:pomsets}
 are syntax-oblivious and they make
minimal assumptions on the communication model.
Therefore, our results can be applied to a wide range of languages.

The closure conditions reviewed in \cref{sec:realisability} have been
initially introduced in \cite{aey03} to study realisability of MSC.
The replacement in the framework of MSC with pomsets is technically
straightforward and yields more general results, since it enables
multi-threaded participants.  In \cref{sec:realisability}, to avoid
systems where participants can get stuck due to the termination of
some partners, we introduce the notion of termination soundness and
demonstrate sufficient conditions that guarantee it.  Then, we
introduce new verification conditions for the distributed
realisability of pomsets, which can tame the combinatorial explosion
due to the interleaving of communication events.

A problem related to realisability is satisfiability of logical formulae.
Model checkers use temporal logic, i.e. LTL, to formalize system
specifications. A general problem that must be faced is
that formal specifications can be wrong as their implementations.
For instance, if a formula is unsatisfiable,
then the specification is probably incorrect.
Similarly to realisability, the problem of satisfiability of a temporal
formula~\cite{rozier2007ltl} allows to demonstrate that there exists an
implementation that meets the specification.

%%% Local Variables:
%%% mode: latex
%%% TeX-master: "main"
%%% End:

%% file: conc.tex
There are some open questions to address.
Pomset semantics of recursive processes is infinite, which precludes
to directly use these results for global specifications that have
loops.
In~\cite{bc88a} pomsets were used in combination with proved
transition systems to give an non-interleaving semantics of CCS;
basically, given a sequence of transitions
$p \xrightarrow{\alpha_1} \cdots \xrightarrow{\alpha_n} q$ between two CCS
processes $p$ and $q$, a pomset $\apom$ can be derived from a
\emph{proved transition system} so that $\apom$ represents the
equivalence class of traces between $p$ and $q$ \quo{compatible} with
traces labelled $\alpha_1, \ldots, \alpha_n$.
This work can help us to generalise our results to infinite computations.

Realisability of high-level MSCs has been addressed in~\cite{loh02},
but the verification conditions are not syntax-oblivious.
The conditions of~\cref{sec:pomsets} are sufficient but not necessary
conditions for realisability.
This is due to the fact that the same semantics (i.e., set of traces)
can be expressed using different sets of pomsets by exploring different
interleavings.
We do not know if a notion of normal forms for families of pomsets can
be used to guarantee that our conditions are necessary.
We conjecture that our semantics could be applied to other
coordination paradigms such as order-preserving asynchronous
message-passing (as the original semantics of CFSMs), synchronous
communications, or tuple based coordination.
We leave the exploration of the robustness of our framework
as future work.
Finally, we plan to extend ChorGram~\cite{chorgram}, a tool we are
currently developing, to implement our theoretical framework and apply
it to the analysis of global specifications.

%%% Local Variables:
%%% mode: latex
%%% TeX-master: "main"
%%% End: